\documentclass[a4paper,twocolumn,10pt,accepted=2025-07-16]{quantumarticle}
\pdfoutput=1

\usepackage{enumerate,appendix}
\usepackage{amsmath,amsthm,amssymb,bm}
\usepackage{color,graphicx}
\usepackage[usenames,dvipsnames,svgnames,table,cmyk,hyperref]{xcolor}
\usepackage[colorlinks]{hyperref}
\usepackage{optidef}
\hypersetup{
	colorlinks = true,
	urlcolor = {blue},
	citecolor = {blue},
	linkcolor= {blue}
}

\usepackage{graphicx}
\usepackage{amsmath}
\usepackage{latexsym}
\usepackage{bbm}

\usepackage{subcaption}
\usepackage{mathrsfs}

\def \be {\begin{equation}}
\def \ee {\end{equation}}

\newcommand{\tr}{\mathrm{Tr}}
\newcommand{\Tr}{\mathrm{Tr}}

\def \Re{\mathrm{Re}\,}
\def \Im{\mathrm{Im}\,}

\def \del{\partial}

\def \cX{{\cal X}}

\def \sofc2{{\cal S}({\mathbb C}^2)}

\def\>{\rangle}
\def\<{\langle}

\if0

\newcommand{\beq}{\begin{equation}}
\newcommand{\eeq}{\end{equation}}
\newcommand{\bqa}{\begin{eqnarray}}
\newcommand{\eqa}{\end{eqnarray}}
\newcommand{\Tr}{\textrm{Tr}}

\newcommand{\forget}[1]{}

\def \cX{{\cal X}}

\def \del{\partial}

\newcommand{\tr}[1]{\mathrm{Tr}\left(#1\right)}
\fi

\newtheorem{theorem}{Theorem}
\newtheorem{lemma}[theorem]{Lemma}

\def\SU{\mathop{\rm SU}}
\def\U{\mathop{\rm U}}

\newcommand{\mh}[1]{#1}
\newcommand{\yo}[1]{#1}

\begin{document}

\title{The Cram\'{e}r-Rao approach and global quantum estimation of bosonic states}

\author{Masahito Hayashi}\email{masahito@math.nagoya-u.ac.jp}
\affiliation{School of Data Science, The Chinese University of Hong Kong, Shenzhen, Longgang District, Shenzhen, 518172, China}
\affiliation{International Quantum AcademyFutian District, Shenzhen 518048, China}
\affiliation{Graduate School of Mathematics, Nagoya University, Nagoya, 464-8602, Japan}

\author{Yingkai Ouyang}
\email{y.ouyang@sheffield.ac.uk}
\affiliation{School of Mathematical and Physical Sciences, University of Sheffield, Sheffield, S3 7RH, United Kingdom}

\begin{abstract}
Quantum state estimation is a fundamental task in quantum information theory, 
where one estimates real parameters continuously embedded in a family of quantum states. 
In the theory of quantum state estimation,
the widely used Cram\'er Rao approach which considers local estimation 
gives the ultimate precision bound of quantum state estimation in terms of the quantum Fisher information.
However practical scenarios need not offer much prior information about the parameters to be estimated, and the local estimation setting need not apply. 
In general, it is unclear whether the Cram\'er-Rao approach is applicable for global estimation instead of local estimation.
In this paper, we find situations where the Cram\'er-Rao approach does and does not work for quantum state estimation problems involving a family of bosonic states in a non-IID setting, where we only use one copy of the bosonic quantum state in the large number of bosons setting.
Our result highlights the importance of caution when using the results of the Cram\'er-Rao approach to extrapolate to the global estimation setting.
\end{abstract}

\maketitle

\section{Introduction}\label{S0}
Quantum sensors promise to estimate parameters with unprecedented precision, and are based on a mathematical primitive known as quantum state estimation.
In quantum state estimation, the task is to estimate physical parameters embedded within quantum states with minimal error.
The Cram\'er-Rao approach \cite{Helstrom,Holevo,two2,two4,HM08,HO}, a prevalent technique in quantum state estimation which provides lower bounds on the minimum mean square error (MSE) of the estimate.
Such lower bounds, known as Cram\'er-Rao bounds (CRBs), 
use Fisher information obtained from quantum measurements.
Since the Fisher information captures only the local structure of a statistical model, 
Cram\'er-Rao bounds are best suited for local estimation problems,
where one assumes that the unknown parameter is within a small neighborhood of a known value. 
In multiparameter quantum state estimation, 
the Cram\'er-Rao approach is more complicated than in the single parameter case; 
unlike the single-parameter case, where the CRB is simply the inverse of the quantum Fisher information \cite{Helstrom}, the multiparameter situation requires additional nontrivial techniques \cite{HO}.

\mh{Despite the prevalence of CRB\yo{s} in quantum state estimation theory, one should note that they are fundamentally designed for local estimation settings, where the parameter's neighborhood is known. 
They are not directly formulated for global estimation settings, where the true parameter's location is unknown and the entire parameter space must be considered \cite{gorecki2020pi,meyer2023quantum,Wojciech,Chesi,Hall_2012,Ha11-2,Ha06}. While the MSE for global estimation is naturally lower bounded by the MSE for local estimation, and thus by a CRB, a CRB truly represents the precision of an estimator only if the MSE for global estimation actually achieves that CRB.}

\mh{To justify this attainability, many researchers traditionally consider the asymptotic setting with independent and identically distributed (IID) quantum states. Indeed, studies like Ref. \cite{two4} have shown that a specific type of CRB (the Holevo type) can be asymptotically attained for global estimation in the IID setting, even for multiple parameters, when the model is fixed. Furthermore, if a state family forms an exponential family, such as thermal states \cite{PhysRevResearch.6.043171,PhysRevLett.127.190402}, the CRB can be attained without requiring an anymptotic limit \cite[Theorem 6.7]{hayashi2016quantum}\cite{nagaoka89selectedpapers,nagaoka87selectedpapers}.}

%The Misconception of IID in Quantum Metrology

\mh{The success of CRBs in these specific state estimation scenarios has led to a widespread expectation that the CRB can also be attained in channel estimation when multiple uses of the same unknown channel are available. This situation is often implicitly treated as an \yo{``}IID setting" for channels within the community. For example, in the canonical phase estimation problem for qubit systems (an estimation of unitary), 
Refs. \cite{Giovannetti,Giovannetti06,Giovannetti11} focused on this problem, demonstrating that the maximum Fisher information achieves Heisenberg scaling and showing that this maximum is attained when the NOON state is used as input. (The NOON state has been experimentally implemented by several studies \cite{Thomas-Peter,Jonathan,Okamoto_2008,Nagata07}.) Building on these findings, 
they concluded that the phase can be estimated with Heisenberg scaling.}
%these papers then claimed that the phase can be estimated with Heisenberg scaling.

\mh{However, this assumption that the CRB is always attainable under "IID" conditions, particularly in global estimation, is often inaccurate. Our paper aims to highlight this critical distinction. For instance, in the canonical phase estimation problem for qubit systems, Ref. \cite{gorecki2020pi} showed that the MSE of the global estimation problem is $\pi^2$ times larger than the MSE of the local estimation problem, whose performance matches Fisher information. Similar discrepancies were observed in Ref. \cite[Section 4]{Ha11-2} and Ref. \cite[Section 5]{Ha06}. Furthermore, Ref. \cite{hayashi2018resolving} found that the NOON state similarly fails to saturate the CRB in the global estimation setting for phase estimation using two bosonic modes. These facts demonstrate the significant difficulty in establishing a direct relationship between the MSE for global estimation and CRBs in channel estimation.}

\mh{It is crucial to note that even the IID setting for state estimation is not universally sufficient to guarantee CRB attainability for global estimation. 
When the complexity of the state family (i.e., the model) increases with the number of available copies, CRBs generally cannot be saturated. For example, in classical systems, CRBs cannot be saturated for estimating entropy and Rényi entropy in the IID setting when the system size increases \cite{Paninski,VV,AOTT}. Similarly, for quantum entropy and quantum Rényi entropies, \yo{quantum CRBs are found to be inaccurate} for global estimation strategies when the system size increases \cite{AISW,WZ,TAD,H24} \footnote{This is because the CRB in these cases equals the varentropy \cite{TAD,H24}, which is upper bounded by $(\log d)^2$ \cite[Lemma 8]{H02}. If the CRB accurately described optimal global estimation performance with a constant error, the sample complexity would be $O(\log d)^2$ for a $d$-dimensional system. However, estimating both classical and quantum entropy to a constant error requires much larger complexity, even in the classical setting \cite{Paninski,VV,AOTT,AISW,WZ}.}}

%Clarifying Conditions for CRB Attainability

\mh{Despite these challenges, it has been shown that the CRB can still be achieved in global estimation settings for channel estimation when the maximum Fisher information increases linearly with respect to the number of uses of the channel \cite{two3}
\footnote{The idea in Ref. \cite{two3} is to consider $m$ uses of a channel as a single effective channel, prepare an optimal input state for it, and then repeat this process $k$ times, viewing it as a state estimation problem. Under these conditions, the CRB can be attained if the maximum Fisher information scales linearly with the number of uses of the channel. Recently, Ref. \cite{PhysRevLett.128.130502} employed this idea to achieve what they term Heisenberg scaling. However, in this case, the inverse of the MSE behaves as $O(km^2)$ while the total number of channel uses is $km$, meaning that true Heisenberg scaling, which demands an $O((km)^2)$ scaling, is not achieved.}
Given the possible disparity between the MSE for global and local estimation problems even in the IID setting, one might not expect CRBs to be accurate in more complicated non-IID settings. Surprisingly, however, in classical estimation theory, it was found that for non-IID situations, such as globally estimating parametrized classical Markovian processes and classical hidden Markovian processes with a finite state system, the MSE is accurately described by CRBs \cite{MH16-8,Hain17-3}.
Nevertheless, the inherent difficulties in guaranteeing CRB attainability have led some researchers to employ complementary approaches like Ziv-Zakai Error Bounds, which offer alternative perspectives on estimation limits \cite{PhysRevLett.108.230401,Rubio_2018}.
Furthermore, as was known even in the physics community \cite{Braunstein_1992}, the CRB cannot be attained with finite samples, even in the classical IID setting with a fixed model.}

Initial steps tackling the attainability of the CRB have been taken. For example, \mh{Ref. \cite{PhysRevResearch.6.L032048,PhysRevA.104.052214}} described a numerical approach based on semidefinite programs that calculates global estimation bounds from local estimation bounds on a fictitious state. Moreover, the attainability of the CRB in the one-copy setting has been studied \cite[Theorem 6.7]{hayashi2016quantum}\cite{nagaoka89selectedpapers,nagaoka87selectedpapers,Nagaoka23}. However, the attainability of the CRB for other non-IID settings has been less studied. Indeed, determining whether a CRB accurately describes global estimation problems remains a non-trivial challenge, even in the IID setting.

In this paper, we study the accuracy of Cram\'er-Rao bounds for the global estimation of bosonic quantum states in a very non-IID setting, where we only have one copy of a parametrised bosonic quantum state, and consider the limit where the number of bosons becomes very large. 
Bosonic quantum states are ubiquitous, because any fundamental particle in the universe is either a boson or a fermion.
Furthermore, we can realize a boson as a composite particle, comprising of an even number of fermions and number of bosons. 
In the mathematical framework of second quantization, 
bosons are naturally represented in the Fock basis, where basis elements count the occupancy of bosons in the available modes.
For indistinguishable bosons, the corresponding quantum state resides within a space spanned that is invariant under any permutation of the underlying bosons, and are hence symmetric. 
We can realize such bosons in various physical systems, such as Bose-Einstein condensates (BECs) in cold atomic systems \cite{anglin2002bose}.

Mathematically, bosonic states have the same structure as 
symmetric states.
In the context of local estimation theory, symmetric states promise a quantum advantage in certain quantum metrology problems in the noiseless setting.
The additional simplicity of the preparation and control of symmetric states \cite{johnsson2020geometric,ouyang2022quantum} makes symmetric states attractive candidates to demonstrate the near-term advantage of quantum technologies. 
However in a practical setting, we may not have the requisite prior information about the parameters embedded within these symmetric states that are to be estimated, 
and require a global estimation strategy.
Given this, it is pertinent to understand the applicability of the Cram\'er-Rao approach, particularly for the symmetric states for which it is purported that 
a quantum advantage might be available.

First, we focus on families of bosonic states that are diagonal in the number basis.
Namely, we consider bosonic states that are probabilistic mixtures of states with a fixed number of bosons with the following probability distributions;
(1) a geometric distribution, (2) a binomial distribution, and (3) a delta distribution.
A notable example of the delta distribution is the half-Dicke state, which has a quantum advantage in the parameter estimation in the direction of SU(2) using Fisher information \cite{PhysRevA.85.022321,tothPhysRevA.85.022322,lucke2011twin,halfdicke}.
The embedded parameter describes the probability distribution.
In this case, we show that the MSE from the Cram\'er-Rao approach is equal to the MSE in the global estimation setting.

Second, we proceed to families generated by 
unitary evolutions of $\SU(2)$
over probe states that begin in the number basis.
Unlike many previous studies \cite{gorecki2020pi,meyer2023quantum,Wojciech,Chesi,
DARIANO1998103,Macchiavello03,Gebhart,
Wojciech22,Hayashi_2022,Luis,
Buzek,Higgins2007,Hall_2012,Kitaev,Cleve,Imai_2009} that focus on estimating mutiple phases in a similar setting, we specifically address the less studied scenario of  
the estimation of the state family generated by
unitary evolutions of $\SU(2)$.
The probe states can be in a (1) binomial distribution, (2) geometric distribution, and a (3) delta distribution in the number basis. 
For this, we consider the problem of estimating parameters embedded in a unitary model, where the unitary channel acts on the probe state.
We analyze the global estimation of this problem by drawing an analogue between the bosonic system and the $\SU(2)$ system with a spin-$j$ system, and we employ the covariant approach initiated by Holevo \cite{HOLEVO1979385}, \cite[Chapter 3]{Holevo}. 
\mh{Fortunately, this covariant approach guarantees that
a covariant estimator achieves the optimal estimator 
with group symmetry \cite{HOLEVO1979385}, \cite[Chapter 4]{Holevo}, \cite[Chapter 4]{Group2}, which allows us to restrict our estimator 
within covariant estimators.}
Since we may describe global estimation using an appropriate minimax problem,
the covariant approach works for the global estimation.
By calculating both local and global estimation bounds, 
we are able to determine if the Cram\'er-Rao approach is accurate for the global estimation of our unitary model. 
For the half-Dicke state, we show that global estimation does not have the quantum advantage that local estimation promises.

The remaining part of this paper is organized as follows.
First in Section \ref{S1}, we explain how bosonic states may arise in practice, 
and one may prepare the families of bosonic states that we consider in our paper.
Second in Section \ref{S2}, we review the general formulation of quantum state estimation based on the Cram\'{e}r-Rao approach.
Third in Section \ref{S3}, 
we discuss the attainability of the Cram\'er-Rao bound in the global estimation setting with respect to several quantum state estimation problems.
Fourth in Section \ref{S4}, we discuss the attainability and the unattainability of the Cram\'er-Rao bound in the global estimation setting 
with respect to several quantum state estimation problems under a unitary model.
Finally in Section \ref{S6}, we have a final discussion of the results that we obtain. 

We reiterate that our paper allows one to calculate the optimal global estimation bound for a unitary channel that acts on a symmetric probe state, 
and we show that such a global bound need not be equal to the local estimation bound from the Cram\'er-Rao approach.

\section{Bosonic states and their preparation}\label{S1}
\subsection{Boson Fock space and geometric distribution}
There are physical systems where we may realize indistinguishable identical bosons. 
For instance, ultracold neutral atoms, when sufficiently cooled and confined, can become indistinguishable, and hence are fundamentally bosonic states. 
We can interpret neutral atoms using their total spin or electronic states as internal degrees of freedom as bosonic states.
Similarly, we can interpret photons that are indistinguishable in all aspects except for their polarizations as bosonic states.

An example of a bosonic system that is controllable in the near term with a large number of bosons is a system of ultracold neutral atoms. 
Neutral atoms can be realised as bosons if we interpret each neutral atom as a composite particle with an equal number of protons and electrons and an even number of neutrons. 
Almost every neutral atom has an isotope that is a boson. 
Examples of neutral atoms that are bosons include group I elements such as 
Li-7, Na-23 and Rb-87. BECs of such indistinguishable identical neutral atoms are now routinely realized in experiments, with the number of bosons being as large as $10^{10}$ \cite{anglin2002bose}.
For ultracold neutral atoms, the internal degrees of freedom can for instance correspond to the total spin of each atom, which can take on two accessible values.
For photons, the internal degrees of freedom can correspond to their horizontal and vertical polarizations.

Our paper considers the quantum state estimation problem for a system of $n$ identical and indistinguishable bosons. 
We model the bosonic system with $d$ kinds of distinguishable modes as the $d$-mode bosonic Fock space ${\cal H}_{B,d}$, which is spanned by the basis $\{ |n_1, n_2, \ldots, n_d\rangle_B : n_1 , \dots, n_d \ge 0\}$.
The space ${\cal H}_{B,d}$ is written as the tensor product space 
${\cal H}_B^{\otimes d}$,
where ${\cal H}_B$ is the one-mode bosonic Fock space 
spanned by the basis $\{ |n\rangle_B\}_{n=0}^{\infty}$.
The $d$-mode Fock space can also decomposed as
\begin{align}
{\cal H}_{B,d}=\bigoplus_{n=0}^{\infty}{\cal H}_{B,d,n},
\end{align}
where ${\cal H}_{B,d,n}$ 
are Fock spaces with a total of $n$ bosons in $d$ modes
spanned by 
\begin{align}
 \mathcal B_n \coloneqq 
 \{|n_1, n_2, \ldots, n_d\rangle_B : \sum_{k=1}^d n_k=n , n_k \ge 0\}.
\end{align}
We may interpret the spaces ${\cal H}_{B,d,n}$ as constant excitation spaces with $n$ excitations \cite{ouyang2019permutation}, which are eigenspaces of Hamiltonians that are sums of independent and identical single-mode operators diagonal in the Fock basis.
The space ${\cal H}_{B,d,n}$ is also isomorphic to the symmetric subspace in $n$-fold tensor product space of the $d$-dimensional space. Since automorphisms on symmetric space can be described using the group $\SU(d)$, we can also use $\SU(d)$ to describe automorphisms on ${\cal H}_{B,d,n}$.

Hereafter, we focus on the case with $d=2$, which corresponds to bosons with two internal degrees of freedom. 
Denoting the spin-$\frac{n}{2}$ space as ${\cal H}_{\frac{n}{2}}$, we may decompose
the space ${\cal H}_{B,d,n}$ as a direct sum of spin-spaces given by
\begin{align}
{\cal H}_{B,2}=\bigoplus_{n=0}^{\infty}
{\cal H}_{\frac{n}{2}}.
\end{align}
The spin-$j$ space is spanned by 
$\{|j ;m\rangle\}_{m=-j}^j$, and its automorphisms have the symmetry of the group $\SU(2)$,
where the operators $J_1,J_2$ and $J_3$ satisfy the commutation relations 
$[J_i,J_j] = \sqrt{-1} \epsilon_{i,j,k} J_k$ and form the Lie algebra of $\SU(2)$, with $\epsilon_{i,j,k}$ denoting the Levi-Civita symbol, and $J_3$ being a diagonal operator in the Fock basis.
Using this idea, we can identify 
the vector $|n-k,k\>_B$ in the boson Fock space with 
the vector 
$|\frac{n}{2}; k-\frac{n}{2}\>$ in the spin-$\frac{n}{2}$ space.

We may realize the geometric distribution on two-mode Fock states in the number basis with a total of $n$ bosons by starting from thermal states of a 
two-mode Hamiltonian given by $G_{\alpha_1,\alpha_2} \coloneqq \alpha_1 N_1 + \alpha_2 N_2$,
where $N_j$ is the number operator 
on the $j$-th mode.
This Hamiltonian $G_{\alpha_1,\alpha_2}$ represents the sum of two independent single-mode Hamiltonians in the Fock basis where the energy properties of the two modes can be different.
The thermal state that corresponds to the Hamiltonian $G_{\alpha_1,\alpha_2}$ at the inverse temperature $\beta$ is 
given as $c\exp(-\beta G_{\alpha_1,\alpha_2})$ for some normalizing constant $c$,
which we can write as
\begin{align}
\rho_{{\rm G}, \alpha_1,\alpha_2, \beta} = c\!\!\!\!
\sum_{n_1, n_2 \ge 0}\!\!
e^{-\beta (\alpha_1 n_1 + \alpha_2 n_2) }
|n_1, n_2\>_{\mathcal B}  \<n_1, n_2|_{\mathcal B} .
\end{align}
After we measure the total number of bosons and observe $n$ bosons, 
the state becomes
\begin{align}
\rho_{{\rm G}, r} ^{(n)}
&= c'
\sum_{k=0}^n
e^{-\beta (\alpha_1 (n-k) + \alpha_2 k) }
|n-k,k\>_{\mathcal B}  \<n-k, k|_{\mathcal B}  \notag\\
&= c''
\sum_{k=0}^n
r^k 
|n-k,k\>_{\mathcal B}  \<n-k, k|_{\mathcal B} 
\end{align}
for other some normalization constants $c'$ and $c''$,
where $r=e^{-\beta(\alpha_1-\alpha_2)}.$
The state 
$\rho_{{\rm G}, r} ^{(n)} $ 
is a geometric distribution in the number of bosons in the second mode, with geometric ratio given by $r$.

Next, we consider the case when a beam splitter operator applies across the two modes.
Since the beam splitter operator corresponds to an element of 
$g \in \SU(2)$,
we can consider the state estimation for the state family
$\{ U_{\frac{n}{2},g}
\rho_{{\rm G},r}^{(n)} 
U_{\frac{n}{2},g}^\dagger : g\in SU(2)\}$
on the spin-$\frac{n}{2}$ space ${\cal H}_{\frac{n}{2}}$,
where $U_{\frac{n}{2},g}$ denotes a unitary representation of $g$ on the space $\mathcal H_{\frac n 2}$.

Since $J_3$ is a diagonal operator in the Fock basis, it leaves the state 
$\rho_{{\rm G},r}^{(n)}$ invariant.
Then we may identify
the parameter space as the homogeneous space $\SU(2)/\U(1)$, 
where $\U(1)$ is the one-parameter group generated by $J_3$. 
We consider estimating the group parameter $[g]\in \SU(2)/\U(1)$
under the state family $\{ U_{\frac{n}{2},g}
\rho_{{\rm G},r}^{(n)}
U_{\frac{n}{2},g}^\dagger : {[g]\in \SU(2)/\U(1)} \}$,
which amounts to estimating two real-valued parameters.

\subsection{Symmetric space and binomial distribution}
Since a spin-$\frac{n}{2}$ system is mathematically equivalent 
to the symmetric subspace of an $n$-qubit system,
we can represent a bosonic state mathematically as a symmetric state on $n$ qubits.
Manipulating symmetric states is achievable in the near-term, because the requisite quantum control techniques do not require the individual addressability of individual qubits \cite{johnsson2020geometric}. 
By leveraging on existing experimental know-how both in creating BECs and controlling large numbers of identical indistinguishable neutral atoms \cite{anglin2002bose} and controlling photonic systems \cite{erhard2020advances}, 
conducting actual quantum sensing experiments on such symmetric states is a near-term possibility. 

In this scenario, 
we can consider another distribution instead of the geometric distribution.
That is, we discuss how to prepare 
the following state over the spin-$j$ system
\begin{align}
\rho = \sum_{m=-j}^j p_m |j;m\>\<j;m|,\label{BNA}
\end{align}
where 
(1) $p_k$ follows a binomial distribution and 
(2) $p_k$ follows a delta distribution.

On the symmetric space, the operators $J_1$, $J_2$ and $J_3$ 
are angular momentum operators that map symmetric states to symmetric states. In terms of the Pauli operators $\sigma_1, \sigma_2$ and $\sigma_3$, 
we can write the angular momentum operator $J_j$ as 
\begin{align}
J_j  = \frac{1}{2}  ( \sigma_j^{(1)} + \dots +  \sigma_j^{(n)} ) ,
\end{align}
where $\sigma_j^{(k)} $ denotes 
Pauli operator $\sigma_j$ on the $k$-th particle. 

One may prepare a quantum state $\rho$ with a binomial distribution of states in the basis $\{|j {;} m\rangle\}_{m=-j}^j$, that is where 
$p_m =  {n \choose \frac{n}{2}+m} p^{\frac{n}{2}-m} (1-p)^{\frac{n}{2}+m}$, according to the following procedure.
First, one prepares the initial separable state
\begin{align}
|\psi_p\> = \Big(\sqrt{1-p}|\frac{1}{2};-\frac{1}{2}\> + \sqrt p |\frac{1}{2};\frac{1}{2}\>\Big)^{\otimes n} .
\end{align}
Second, one dephases the pure state $|\psi_p\> $ in the basis $\mathcal B_n$ using the master equation $d \tau / dt = {\mathcal D}(\tau) $, 
where 
\begin{align}
{\mathcal D}(\rho) = 
\gamma ( J_3 \rho J_3^\dagger - \frac{1}{2}J_3^\dagger J_3\rho - \frac{1}{2}\rho  J_3^\dagger J_3) .
\end{align}
Since we have
\begin{align}
|\psi_p\>= 
\sum_{m=-\frac{n}{2}}^{\frac{n}{2}} \sqrt{p^{\frac{n}{2}-m} (1-p)^{\frac{n}{2}-m} } \sqrt{  {n \choose \frac{n}{2}+m} }
|\frac{n}{2} ; m\>,
\end{align}
complete dephasing of the state $|\psi_p\>$ in the eigenbasis of $J_3$ will yield a binomial distribution of states in the basis 
$\{|j ; m\rangle\}_{m=-j}^j$. 
In particular, we have $\lim_{t \to \infty} e^{{\mathcal D} t}(|\psi_p\>\<\psi_p|) = 
\rho^{(n)}_{{\rm B},p}$ where 
\begin{align}
\rho^{(n)}_{{\rm B},p}:=\!\!\!
\sum_{m=-\frac{n}{2}}^{\frac{n}{2}}%, n_2: n_1+n_2=n } \!\!\!
{n \choose \frac{n}{2}+m}
p^{\frac{n}{2}-m} (1-p)^{\frac{n}{2}+m}  
 |\frac{n}{2};m\>\<\frac{n}{2};m|,
\end{align}
and
$ e^{{\mathcal D} t} = \mathcal I + \sum_{k=1}^\infty \frac{t^k}{k!} {\mathcal D}^k$, and $\mathcal I$ denotes the identity operator.
Hence, one can apply $e^{t {\mathcal D}}$ on $|\psi_p\> $ for large $t$ to approximately obtain a binomial distribution on states in the basis $\{|j{;}m\rangle\}_{m=-j}^j$. 

Once an unknown unitary $U_{\frac{n}{2},g}$ with $g \in \SU(2)$
is applied,
in the same way as with the geometric distribution, 
we can consider the state family 
$\{ U_{\frac{n}{2},g}
\rho^{(n)}_{{\rm B},p}
 U_{\frac{n}{2},g}^\dagger: {[g] \in \SU(2)/U(1)} \}$.

\if0
One can also obtain a geometric distribution of states in the basis $\mathcal B_n$ using 
the master equation $d\tau /dt = \mathcal D (\tau)$
where
\begin{align}
\mathcal D(\rho) = 
a_+( J_+ \rho J_+^\dagger - \frac{1}{2}J_+^\dagger J_+\rho - \frac{1}{2}\rho  J_+^\dagger J_+)
+
a_- ( J_- \rho J_-^\dagger - \frac{1}{2} J_-^\dagger J_-  \rho - \frac{1}{2} \rho J_-^\dagger J_- ).
\end{align}
Here, we have $a_+,a_- \in \mathbb R$ and $J_\pm \coloneqq \frac{1}{2}(J_1 \pm i J_2)$, and $a_+ \neq a_-$.
A geometric distribution of states in the basis $\mathcal B_n$
has probability distribution satisfying the equations $p_{w+1} = r p_w $ for all $w=0,\dots, n-1$ for some positive number $r$,
and is a fixed point of the operator $\mathcal D$.
In particular, for any symmetric state $\tau$, we have 
\begin{align}
\lim_{t \to \infty} e^{t \mathcal D} \tau = 
\rho^{(n)}_{{\rm G},r}:=
\sum_{k=0}^{n}%, n_2: n_1+n_2=n }
\frac{r-1}{r^{n+1}-1}r^{k}
|\frac{n}{2};m\>\<\frac{n}{2};m|,
\end{align}
where $r = a_+/a_- \neq 1$.
We prove this in Lemma \ref{lem:geometric} of the Appendix.
Hence, in a physical system with Markovian dynamics where both coherent gain and loss errors occur, $\mathcal D$ governs the dynamics of the physical system and any symmetric state will approach this geometric distribution of states in the basis $\mathcal B_n$ for large $t$. 
\fi

\subsection{Delta distribution}
A state with the delta distribution in \eqref{BNA} can be prepared as follows.
There are probabilistic approaches to prepare a specific state in the number basis. 
For the probabilistic approach, one can prepare a binomial or geometric distribution of states in the basis $\{|j{;}m\rangle\}_{m=-j}^j$, and 
subsequently measure in the basis $\{|j{;}m\rangle\}_{m=-j}^j$. 
For the deterministic approach, one can use an ancillary bosonic mode along with a dispersive interaction Hamiltonian that is proportional to $a^\dagger a \otimes J_3$ to implement unitary operations in the 
spin-$j$ system
using geometric phase gates \cite{Alfredo-Luis_2001,zanardi-PhysRevA.65.032327,johnsson2020geometric}.

\section{General formulation of Cram\'{e}r-Rao approach}
\label{S2}
In quantum state estimation, we are given copies of an unknown state $\rho_{\theta_0}$ from the set of quantum states $\{\rho_\theta : \theta = (\theta^1, \dots , \theta^d) \in \Theta \}$
where $\Theta$ is a continuous set in $\mathbb R^d$.
We assume that the quantum states $\rho_\theta$ are differentiable with respect to parameter $\theta$ for all $\theta \in \Theta$. 
Our objective is to find the minimum MSE of a locally unbiased estimator $\hat \theta$ that estimates the true parameter $\theta_0$.

We describe a measurement using a set of positive operators $\Pi = \{ \Pi_x : x\in \mathcal X\}$ labeled by a set $\mathcal X$, where the completeness condition $\sum_{x \in \mathcal X}\Pi_x = I $ holds.
By Born's rule, a measurement $\Pi$ on a quantum state $\rho_\theta$ gives the classical label $x$ and the state 
$\Pi_x \rho_\theta / \tr(\Pi_x \rho_\theta )$ 
with probability $p_\theta(x) = \tr(\Pi_x \rho_{\theta})$.
Given a function $f$ of the classical label $x$,
 we denote $\mathbb E_{\bm{\theta}}[f(x) |\Pi]$ as the expectation of $f(x)$, with probability distribution obtained according to Born's rule.

Given a measurement $\Pi$ and an estimator $\hat{\bm{\theta}}$ that depends on the classical label $x$, we denote  
$\hat{\Pi}=(\Pi , \hat{{\bm{\theta}}})$ as an {\it estimator}. 
When the true parameter $\theta_0$ is equal to $\theta$, we define the mean-square error (MSE) matrix for the estimator $\hat{\Pi}$ as  
\begin{align}\nonumber
V_{{\bm{\theta}}}[\hat{\Pi}] 
&=
\sum_{i,j=1}^d |i\>\<j|
 \mathbb E_{\bm{\theta}}\big[({\hat{\theta}^i}(x)-\theta^i)({\hat{\theta}^j}(x)-\theta^j)|\Pi\big] \notag.
\end{align}
In multiparameter quantum metrology, the objective is to 
find an optimal estimator $\hat{\Pi}=(\Pi , \hat{{\bm{\theta}}})$ that minimizes $\tr G V_{{\bm{\theta}}}[\hat{\Pi}]$, where a weight matrix $G$, a size $d$ positive semidefinite matrix, quantifies the relative importance of the different parameters. 

Our estimator $\hat{\Pi}$ is unbiased at $\theta_0 = \theta$ if for all $i = 1,\dots, d$, the expectation of our estimator equals the true value of the parameter $\theta_0$ when $\theta_0 = \theta$, that is 
\begin{align}
\mathbb E_{\bm{\theta}}\big[{\hat{\theta}^i}(x)|\Pi\big]&
=\sum_{x\in\cX} {\hat{\theta}^i}(x) \tr{\big[\rho_{{\bm{\theta}}}\Pi_x\big]}
=\theta^i \label{MK}.
\end{align}
Our estimator is globally unbiased if \eqref{MK} holds for all $\theta \in \Theta$.
We can also consider locally unbiased estimators, which are estimators that are unbiased in the neighborhood of the true parameter $\theta_0$.  
For this aim,
we define $D_j:= \frac{\partial\rho_\theta}{\partial \theta^j}|_{\theta=\theta_0}$,
and $\rho:= \rho_{\theta_0}$.
Taking partial derivatives on both sides of \eqref{MK}, we get
\begin{align}
\frac{\del}{\del\theta^j}\mathbb E_{\bm{\theta}}\big[{\hat{\theta}^i}(x)|\Pi\big]&=
\sum_{x\in\cX} {\hat{\theta}^i}(x)\tr{   D_j \Pi_x}
=\delta_i^j \label{M1}.
\end{align}
The estimator $\hat{\Pi}$ is locally unbiased if \eqref{MK} holds for all $i=1,\dots, d$ for a fixed $\theta$ where $\theta_0 = \theta$, and when \eqref{M1} holds for all $i,j = 1,\dots, d$.

 For any weight matrix $G = \sum_{i,j=1}^d g_{i,j}|i\>\<j|$, 
the tight Cram\'{e}r-Rao (CR) type bound, i.e.,
 the fundamental precision limit \cite{HO}, is
\be\label{qcrbound}
C_{\bm{\theta}}[G]:=
\min_{\hat{\Pi}\mathrm{\,:l.u.at\,}{\bm{\theta}}}\Tr{ \big[G V_{\bm{\theta}}[\hat{\Pi}]\big]}, \notag
\ee
where `l.u.~at ${\bm{\theta}}$' indicates our minimization over all possible estimators under the locally unbiasedness condition. 
Since this minimum is attained by $\hat{\Pi}$ satisfying \eqref{MK} when we impose only the condition \eqref{M1}, it suffices to consider $C_{\bm{\theta}}[G]$ as a minimization with only the condition \eqref{M1}.

To evaluate $C_{\bm{\theta}}[G]$, 
we often focus on the symmetric logarithm 
derivative (SLD) $L_j$, which is an operator that satisfies the equation
\begin{align}
D_j= \frac{1}{2} (L_j \rho +\rho L_j).
\end{align}
The SLD Fisher information matrix $F=(F_{i,j})$ is given as
\be
F_{i,j}:= \frac{1}{2}  \Tr L_i (L_j \rho +\rho L_j).
\ee
The tight CR bound $C_{\bm{\theta}}[G]$ can be lower bounded as follows 
\be
C_{\bm{\theta}}[G] \ge C^{\rm S}_{\bm{\theta}}[G]:=\Tr G F^{-1}.\label{CZO}
\ee
The RHS of \eqref{CZO} is called the SLD bound.

In the one-parameter case, 
we do not need to handle the trade-off among various parameters.
In this case, the equality in \eqref{CZO} holds.
We can attain this bound using a projective measurement in the eigenbasis of the SLD $L$.
Hence, in the multiple-parameter case, 
when the SLDs $L_j$ are non-commutative,
their spectral decompositions cannot be measured simultaneously.
However, it is possible to randomly choose one of the SLDs $L_j$ and measure it, as was studied in \cite{H97}.
To discuss a simple case of this strategy,
we assume that the SLD Fisher information matrix $J$ has no off-diagonal element.
%When $G=I$,
The tight CR bound $C_{\bm{\theta}}[G]$ can be evaluated simply as follows \cite{Helstrom}.
\be
C_{\bm{\theta}}[G] \le d \Tr G F^{-1}.
\ee
We attain the SLD bound by measuring in the eigenbasis of the SLD $L_j$ with equal probability for $j=1, \ldots, d$.
Thus, when $d=2$, the SLD bound decides
$C_{\bm{\theta}}[G]$ within twice the range.

To get a better lower bound, 
we often focus on the right logarithm 
derivative (RLD) $\tilde{L}_j$, which is an operator that solves the equation
\begin{align}
D_j= \rho \tilde{L}_j.
\end{align}
The RLD Fisher information matrix $\tilde{F}=(\tilde{J}_{i,j})$ is given as
\be
\tilde{F}_{i,j}:=  \Tr \tilde{L}_i \rho \tilde{L}_j.
\ee
The tight CR bound $C_{\bm{\theta}}[G]$ can be lower bounded as follows.
\begin{align}
C_{\bm{\theta}}[G] \ge 
C^{\rm R}_{\bm{\theta}}[G]:=&\Tr 
\Re \sqrt{G} \tilde{F}^{-1}\sqrt{G}
\notag\\
&+\Tr
|\Im \sqrt{G} \tilde{F}^{-1}\sqrt{G}|
.\label{CZO4}
\end{align}
The RHS of \eqref{CZO4} is called the RLD bound \cite[Chapter 6]{Holevo}.
%Now, we assume that
To consider this bound, we define the operator $D$ as
\begin{align}
[\rho,X]=\frac{1}{2}(\rho D(X)+D(X)\rho).
\end{align}
We say the D-invariant condition holds if $D(L_j)$ is in the linear span of $L_1, \ldots, L_d$. 
We define matrix $D=(D_{j,k})$ as
$D_{j,k}:= \Tr D(L_j) D_k $.
In this case,
%When the model satisfies a certain condition, 
the $\tilde{F}^{-1}$ is calculated as \cite[Chapter 6]{Holevo}
\begin{align}
\tilde{F}^{-1}= F^{-1}+\frac{i}{2} F^{-1}D F^{-1}.\label{NMN}
\end{align}
Then, the RLD bound is calculated as \cite[Chapter 6]{Holevo}
\begin{align}
C^{\rm R}_{\bm{\theta}}[G]=\Tr G F^{-1}+\frac{1}{2}\Tr |\sqrt{G} F^{-1}D F^{-1} \sqrt{G}|.
\end{align}
and is
a better lower bound than the SLD bound.
Since \eqref{NMN} implies 
\begin{align}
\tilde{F}^{-1} \le 2 F^{-1},
\end{align}
we have 
\begin{align}
\Tr G F^{-1} \le C^{\rm R}_{\bm{\theta}}[G]\le 2 \Tr G F^{-1}. \label{LOA}
%=2 C^{\rm S}_{\bm{\theta}}[G].
\end{align}
That is, the RLD bound differs from the SLD bound by up to a factor of two for D-invariant models.

As a tighter lower bound, we employ Holevo-Nagaoka (HN) bound as follows
\cite{Holevo,nagaoka89-2,HM08}.
Given a tuple of Hermitian matrices $\vec{X}=(X_1, \ldots,X_d)$, we define 
the matrix $Z(\vec{X})=(Z_{j,k}(\vec{X}))$ as
\begin{align}
Z_{j,k}(\vec{X}):= \Tr \rho X_j Z_k.
\end{align}
We impose the following condition to $\vec{X}$;
\begin{align}
\Tr X_j D_k=\delta_{j,k} \label{BNO}.
\end{align}
Then, we define
\begin{align}
C^{\rm HN}_{\bm{\theta}}[G]:=
\min_{\vec{X}} \Tr G \Re Z(\vec{X})+
\Tr |\sqrt{G} \Im Z(\vec{X})\sqrt{G}|,
\end{align} where
the minimum is taken under the condition \eqref{BNO}.
Then, we have
\begin{align}
C^{\rm HN}_{\bm{\theta}}[G] \le C_{\bm{\theta}}[G].
\end{align}
Furthermore, we have
\begin{align}
C^{\rm R}_{\bm{\theta}}[G] &\le C^{\rm HN}_{\bm{\theta}}[G] \label{NAS}\\
C^{\rm S}_{\bm{\theta}}[G] &\le C^{\rm HN}_{\bm{\theta}}[G].
\end{align}
When the model is D-invariant, the equality in \eqref{NAS} holds \cite{HM08}.

For example, when we choose $\vec{X}$ as
$\vec{X}_*=(X_{k,*})$ with
$X_{k,*}:= \sum_{j=1}^d (F^{-1})_{k,j} L_j$,
$\vec{X}_*$ satisfies the condition \eqref{BNO}.
Also, we have $\Re Z(\vec{X}_*)=F^{-1}$.
Since $Z(\vec{X})\le 2 \Re Z(\vec{X})$,
we have
\begin{align}
C^{\rm HN}_{\bm{\theta}}[G] 
&\le
\Tr G \Re Z(\vec{X}_*)+
\Tr |\sqrt{G} \Im Z(\vec{X}_*)\sqrt{G}| \notag\\
&\le 2 \Tr G \Re Z(\vec{X}_*) = 2\Tr G F^{-1}.
\end{align}
That is, the HN bound differs from the SLD bound by up to a factor of two 
for D-invariant models.
Hence for D-invariant models, we have good upper and lower bounds on the tight CR-bound based on the easily computable SLD bound.

\if0
This bound is another form.
We expand our model around the state $\rho=\rho_{\theta_0}$
so that we have additional derivatives $D_{d+1}, \ldots, D_{d'}$
around the state $\rho$.
The expanded model is assumed to be D-invariant.
We denote the RLD bound in the expanded model by
$C^{\rm R}_{\bm{\theta},ex}[G']$.
Then, we introduce a $d'\times d$ matrix $P$ under the following condition;
\begin{align}
P_{k,j}=\delta_{k,j}. \label{ZXP}
\end{align}
Then, we have
\begin{align}
C^{\rm HN}_{\bm{\theta}}[G]=
\min_P C^{\rm R}_{\bm{\theta},ex}[ P GP^T],
\end{align}
where the minimum is taken under the condition \eqref{ZXP}.
When we denote the SLD Fisher information matrix of the extended model
by $F_{ex}$,
using \eqref{LOA}, we have
\begin{align}
\min_P C^{\rm R}_{\bm{\theta},ex}[ P GP^T]
\le \min_P 2 \Tr F_{ex}^{-1} P GP^T 
= \min_P 2 \Tr P^T F_{ex}^{-1} P G
\le 
\end{align}
\fi

\section{Attainability of the Cram\'er-Rao bound in the global estimation setting}\label{S3}
 
Here, we give examples where the tight Cram\'er-Rao bound equals to the minimum MSE for global estimation strategies.

First, we consider the task of estimating the parameter $p$ in the state $\rho^{(n)}_{{\rm B},p}$ that is a binomial distribution of states in the basis 
$\{|n-k,k\>_{B}\}_{k=0}^n$.
By measuring in the basis $\{|n-k,k\>_{B}\}_{k=0}^n$, 
we obtain a binomial distribution, 
which has a Fisher information of $\frac{2j}{p(1-p)}$. Hence the tight Cram\'er-Rao bound $C_{p}[1]$ is $\frac{p(1-p)}{2j}$.
We can attain this bound with a global estimator of the parameter $p$ according the following strategy. 
First we measure this density matrix in the basis 
$\{|n-k,k\>_{B}\}_{k=0}^n$.
Second, if we observe the state $|n-k,k\>_{B}$, we set our estimate as $\frac{k}{n}$.
This estimator is unbiased because it has expectation $p$. Moreover, it has MSE $\frac{p(1-p)}{n}$ which attains the tight Cram\'er-Rao bound.

Second, we consider estimating the parameter $r$ in the state 
$\rho^{(n)}_{{\rm G},r}$ which is a normalized geometric distribution on states in the basis $\mathcal B_n$.
By measuring in the basis $\mathcal B_n$, the estimation problem reduces to estimating a geometric distribution.
Now, let us see why the tight Cram\'er Rao bound in this case is equal to the minimum MSE for the global estimation of $r$. 

We begin with the parametrization 
$P_\theta(k):=
\frac{e^{\theta}-1}{e^{\theta (n+1)}-1}e^{\theta k}$, 
which is known as the {\em natural parameter} in the field of information geometry \cite{AN}.
Since the geometric distribution is an exponential family,
the tight CR bound is globally achieved under the {\em expectation parameter} $\eta(\theta)$ \cite{AN}, which is defined as $\eta(\theta):= \sum_{k=0}^n k P_\theta(k)$ where we may calculate $\eta(\theta)$ as 
\begin{align}
\eta(\theta)  
&= \frac{n e^{\theta (n+1)}+1}{e^{\theta (n+1)}-1}-\frac{1}{e^{\theta}-1} %\notag\\
%&
= \frac{n r^{n+1}+1}{r^{n+1}-1}-\frac{1}{r-1}.
\end{align}
Then, the Fisher information for $\theta$ is $F_\theta:= 
\sum_{k=0}^n k^2 P_\theta(k) - \eta(\theta)^2$ which can be calculated as
\begin{align}
&F_\theta 
= \frac{n(n-1)e^{\theta (n+1)}+2}{e^{\theta (n+1)}-1}
+\frac{e^{\theta (n+1)}-3}{e^{\theta (n+1)}-1}
\eta(\theta)- \eta(\theta)^2 \notag \\
=& \frac{n(n-1)r^{n+1}+2}{r^{n+1}-1}
+\frac{r^{n+1}-3}{r^{n+1}-1}
- \eta(\theta)^2 \notag \\
=& \frac{n(n-1)r^{n+1}(r^{n+1}-1)+2(r^{n+1}-1)}{(r^{n+1}-1)^2}\notag\\
&+\frac{(r^{n+1}-3)(n r^{n+1}+1)}{(r^{n+1}-1)^2}
\eta(\theta)- \eta(\theta)^2 .
\end{align}
In this case, when the parameter to be estimated is set to $\eta(\theta)$,
the estimator is given as $k$.
This estimator satisfies the unbiasedness condition, and its variance is $F_\theta$, i.e.,
the Fisher information of the natural parameter.
We can use this procedure to estimate $r$ globally with MSE that attains the tight CR bound $C_{\bm{\theta}}[I]$.

\section{Unattainability of the Cram\'er-Rao bound in the global estimation estimation setting}\label{S4}

\subsection{Local estimation of a unitary channel}\label{local:estimation:S4}
We consider the covariant model on symmetric states of $n$ qubits.
Using the representation theory of $\SU(2)$, we interpret such symmetric states with a spin $j=\frac{n}{2}$ system, wherein it is natural to interpret the number state $|n-k,k\>_{\mathcal B}$ as a spin state $|j;-j+k\rangle$.
We focus on a diagonal state $\rho$ for this basis given as
\be
\rho:=\sum_{m=-j}^{j}
p_{m}
 |j;m\rangle  \langle j;m|.
\ee
Then, given a parameter $\theta:=(\theta_1,\theta_2)$,
we consider the state family 
$\rho_{\theta}:= U_{\theta} \rho
U_{\theta}^\dagger $, where
$U_{\theta} := \exp(i (\theta_1 J_1+\theta_2 J_2))$.
The two-parameter space $\Theta$ is given as
$\{ \theta | |\theta| \le \pi\}$.
This state family $\{\rho_{\theta}\}_{\theta}$
has two parameters,
and is obtained from applying unitary operator $U_\theta$ on an initial probe state $\rho$.
In this model, the SLD Fisher information is diagonal, in the sense that $F_{1,2}=F_{2,1} = 0$.
This implies that $C^{\rm S}_\theta[I] = F_{1,1}^{-1} + F_{2,2}^{-1}$.
Furthermore,
\begin{align}
F_{1,1}=F_{2,2}
=\!\!\!\!
\sum_{m=-j}^{j-1}  \!
\frac{4(p_{m+1}-p_m)^2}{p_{m+1}+p_m }
(j-m)(j+m+1) .\label{NMI}
\end{align}
From \eqref{NMI}, we can apply the Cram\'er-Rao approach on probe states initialised as (1) a binomial distribution of number states, (2) a geometric distribution of number states, and (3) a delta distribution of number states.
\begin{enumerate}
\item {\bf Binomial distribution:-}
Now consider the case when $\rho = \rho^{(n)}_{{\rm B},p}$.
When $p$ is fixed and $j=\frac{n}{2}$ increases,
the diagonal element of the SLD Fisher information $F^{(n),{\rm B},p}$
can be calculated as
\begin{align}
F^{(n),{\rm B},p}_{1,1}=F^{(n),{\rm B},p}_{2,2}\cong 2n .\label{BZV}
\end{align}
Hence $C^{\rm S}_\theta[I] \cong 1/n$.
\item {\bf Geometric distribution:-}
Consider $\rho$ as $\rho^{(n)}_{{\rm G},r}$.
From Appendix \ref{A1}, this model satisfies the D-invariant condition.
Hence, the RLD bound gives a tighter lower bound than  
the SLD bound. As calculated in Appendix \ref{A-D},
when $r$ is fixed and $j=\frac{n}{2}$ increases,
the RLD bound is approximated as
\begin{align}
C^{\rm R}_{\bm{\theta}}[I] \cong \frac{4r}{n(r-1)}.
\label{BZV5}
\end{align}

\item {\bf Delta distribution:-}
Consider $p_m=\delta_{m,a}$ for some integer $a \in [-j+1,j-1]$.
Then we have 
\begin{align}
F_{1,1} = F_{2,2}   
=2(j^2-a^2) + 2j.
\end{align}
Hence when 
$a$ is proportional to $n$, both $F_{1,1}$ and $F_{2,2}$ are quadratic in $j$ and $n$.
Then we have
$C^{\rm S}_\theta[I] \cong (1/4-\alpha^2)^{-1}/n^2$ where $\alpha = a/n$.
\end{enumerate}

\subsection{A group covariant approach for global estimation}
\label{S:groupcovarianttheory}
We consider the state family $\{\rho_{\theta}\}_{\theta}$ as given in 
Section \ref{S4}. Since the state family $\{\rho_{\theta}\}_{\theta}$
has a group covariant structure, 
we can employ a group covariant approach \cite{HOLEVO1979385}, \cite[Chapter 4]{Holevo}, \cite[Chapter 4]{Group2}, where we employ a group covariant error function.
\mh{This is because a covariant POVM realizes the optimal performance 
under the symmetric setting
\cite{HOLEVO1979385}, \cite[Chapter 4]{Holevo}, \cite[Chapter 4]{Group2}.}

Now we consider the spin $j$ system ${\cal H}_j$ spanned by 
$\{|j;m\rangle\}_{m=-j}^j$.
For an unknown value of $\theta$, and our estimate $\hat \theta$,
we employ the fidelity between the states
$U_{\theta} |\frac{1}{2};\frac{1}{2}\rangle$
and 
$U_{\hat \theta} |\frac{1}{2};\frac{1}{2}\rangle$
as 
\begin{align}
R(\theta,\hat{\theta})&:=
\tr 
\Big(U_{\theta} |\frac{1}{2};\frac{1}{2}\rangle \langle \frac{1}{2};\frac{1}{2}| U_{\theta}^\dagger\Big)
\Big(U_{\hat{\theta}} |\frac{1}{2};\frac{1}{2}\rangle \langle \frac{1}{2};\frac{1}{2}| U_{\hat{\theta}}^\dagger\Big)\notag\\
&=|\langle \frac{1}{2};\frac{1}{2}| U_{\theta}^\dagger U_{\hat{\theta}} |
\frac{1}{2};\frac{1}{2}\rangle|^2. \label{fidelity}
\end{align}

Using the parameter $\phi$ as
$e^{-i\phi} |\theta|=\theta_1+i\theta_2$, as shown in Appendix \ref{AA1},
the fidelity is calculated as
\begin{align}
R(\theta,\hat{\theta})
=&\Big|
\cos \frac{|\theta|}{2} \cos \frac{| \hat{\theta}|}{2}
+e^{i(-\phi +\hat{\phi})} \sin \frac{|\theta|}{2}\sin 
\frac{|\hat{\theta}|}{2}
\Big|^2.\label{fidelity2}
\end{align}
In particular, 
when $\theta=0$, this fidelity simplifies to
$R(0,\hat{\theta})=
\cos^2 (|\hat{\theta}|/ 2)$.
%Using $e^{-i\phi} |\theta|=\theta_1+i\theta_2$, 
Then, given a parameter $\theta:=(\theta_1,\theta_2)$,
we consider the state family 
$\rho_{\theta}:= U_{\theta} \rho
U_{\theta}^\dagger $.

We define the error function of our estimate to be 
$\eta(\hat \theta) \coloneqq 4(1-R(0,\hat{\theta}))$, which to leading order in $\hat \theta$ is equivalent to $|\hat{\theta}|^2$.
In this definition, the error function $\eta(\hat \theta)$ is to leading order in $\hat \theta$ equivalent to the MSE of $\hat \theta$ in the local estimation scenario when $\theta= 0$. 
Since the fidelity is unitarily invariant,
the error function $\eta(\hat \theta)$
satisfies the group covariance condition required for the group covariant approach.

We identify the parameter space with the homogeneous space $\Theta = \SU(2)/\U(1)$, where $\U(1)$ is the one-parameter group generated by $J_3$. Then, we employ the invariant probability measure $\nu$ on our parameter space $\Theta$
under the above identification.
The estimator is a POVM $M = \{ M(d \hat \theta) : \hat \theta \in \Theta\}$ with outcomes parametrised according to elements in 
$\Theta$.
Given an estimator $M$,  
we focus on the Bayesian average
\begin{align}
R_\nu(M):=
\int_\Theta
\int_\Theta R(\theta,\hat{\theta}) \tr \rho_\theta M(d \hat{\theta})
\nu(d \theta).
\end{align}
Also, we can calculate the performance of the global estimation strategy for the worst possible value of the true parameter using the expression
\begin{align}
R(M):=
\min_{\theta\in \Theta}
\int_\Theta R(\theta,\hat{\theta}) \tr \rho_\theta M(d \hat{\theta}).
\label{BNVX}
\end{align}
Namely, the expression $\eta(M) \coloneqq 4(1-R(M))$ is our error function maximized over all values of the true parameter $\theta$.
Minimizing $\eta(M)$ amounts to solving a minimax problem;
we minimize over all POVMs and maximize over all possible true values of $\theta$.

As before, we denote the representation of $g \in \SU(2)$ on a spin-$j$ system ${\cal H}_j$ by $U_{j,g}$.
We say that the POVM $M$ is covariant if 
\begin{align}
U_{j,g}\int_B M(d \hat{\theta}) U_{j,g}^\dagger=
\int_{g B} M(d \hat{\theta}) 
\end{align}
for any subset $B\subset \Theta$ and $g \in \SU(2)$.
When a state $T$ satisfies the condition
$U_{j,g}TU_{j,g}^\dagger =T$ for $g \in \U(1) $,
we define 
the covariant POVM $M_T$ as
\begin{align}
M_T(B):= (2j+1) \int_{B} U_{j,g(\theta)}TU_{j,g(\theta)}^\dagger \nu(d \theta)
\label{BNNH}
\end{align}
for $B \subset \Theta=\SU(2)/\U(1)$,
where $g(\theta)$ is a representative element of 
$\theta \in \SU(2)/\U(1)$.
Any covariant POVM can be written as the above form.
For our situation, covariant POVMs $M_T$ have $T$ as an operator that is diagonal in the Fock basis.
If $B$ represents an infinitesimal ball about some $b\in \Theta$ and if $T$ represents a pure state $|\phi\>\<\phi|$, 
then $M_T(B)$ is proportional to the operator 
$ U_{j,g(b)}TU_{j,g(b)}^\dagger $,
which in corresponds to a projector onto the state $U_{j,g(b)} |\phi\>$.
Physically, the measurement of a covariant POVM $M_{|\phi\>\<\phi|}(B)$ corresponds to a projection onto the states $U_{j,g(\theta)}|\psi\>$ according to the measure $\nu$.

While the function $R(M)$ is more difficult to calculate than $R_\nu(M)$, 
the situation simplifies greatly 
when the optimal POVM $M$ is covariant.
In this situation,
the Bayesian error function can be equal to the worst-case error function in the sense that 
$1-R_\nu(M)=1-R(M)$ \cite{HOLEVO1979385}, \cite[Chapter 4]{Holevo}, \cite[Chapter 4]{Group2}.
This situation is possible when $R(g\theta,g\hat{\theta})=R(\theta,\hat{\theta})$ and when $\nu$ is invariant under any $g$.
In our scenario the POVM that maximizes $R_\nu(M)$ ($R(M)$)
is realized by a covariant POVM \cite{HOLEVO1979385}, \cite[Chapter 4]{Holevo}, \cite[Chapter 4]{Group2}.
In such a situation, we can calculate 
$R(M)$ as well as $R_\nu(M)$.

To minimize the error function, we use the idea of the addition of a spin-1/2 particle to a spin-$j$ particle,
\begin{align}
{\cal H}_{\frac{1}{2}} \otimes {\cal H}_j=
{\cal H}_{j+\frac{1}{2}} \oplus {\cal H}_{j-\frac{1}{2}},\label{BBY}
\end{align}
where $\mathcal H_j$ denotes the space for spin $j$.
We denote the projection to 
${\cal H}_{j+\frac{1}{2}} $ and $ {\cal H}_{j-\frac{1}{2}}$
by $P_{j+\frac{1}{2}} $ and $P_{j-\frac{1}{2}} $, respectively.

\begin{theorem}[\protect{\cite[Theorem 4.6]{Group2}}]\label{TH2}
When the relation
\begin{align}
&\frac{1}{2j+2}
\tr P_{j+\frac{1}{2}}
|\frac{1}{2};\frac{1}{2}\rangle \langle \frac{1}{2};\frac{1}{2}|\otimes \rho
\notag\\
\ge&
\frac{1}{2j}
\tr P_{j-\frac{1}{2}}
|\frac{1}{2};\frac{1}{2}\rangle \langle \frac{1}{2};\frac{1}{2}|\otimes \rho
\label{BFY}
\end{align}
holds under the relation \eqref{BBY},
the maximum of $R(M)$
%the maximum average of $R(\theta,\hat{\theta})$ 
is 
\begin{align}
\frac{2j+1}{2j+2}
\tr P_{j+\frac{1}{2}}
|\frac{1}{2};\frac{1}{2}\rangle \langle \frac{1}{2};\frac{1}{2}|\otimes \rho,
\end{align}
and the optimal measurement is given as $M_{|j;j\rangle \langle j;j|}$,
which is defined in \eqref{BNNH}, and comprises of a measure on pure states with maximum total angular momentum.
\end{theorem}
For reader's convenience, 
we give a proof for Theorem \ref{TH2} in Appendix \ref{A-F}
as a special case of \cite[Theorem 4.6]{Group2}.
The proof of Theorem \ref{TH2} uses the following ideas.
First, we represent the error function using a spin-1/2 representation 
while considering the estimation of a quantum state in the spin-$j$ representation.
Second, we apply the definition of $R(M)$ and use the fact that the product of traces is equal to the trace of the tensor products of the arguments,
and use Schur's lemma appropriately.

\if0
The optimal estimator
\begin{align}
\int_D \hat{\theta}_1
\Tr U_{\hat \theta} |\frac{1}{2};\frac{1}{2}\rangle
\langle \frac{1}{2};\frac{1}{2}| U_{\hat \theta}^\dagger
\Big[ |\frac{1}{2};\frac{1}{2}\rangle \langle \frac{1}{2};\frac{1}{2}|,
J_1\Big]
\end{align}
\fi

\subsection{Global estimation of a unitary channel}\label{Global-estimation}
Here we apply the theory reviewed in 
Section \ref{S:groupcovarianttheory} to calculate the minimum MSE for the global estimation of our unitary model.

\subsubsection{Probe state as a binomial distribution in the number basis}
First, we consider whether the condition \eqref{BFY} holds
when $p_m$ is given as a binomial distribution and $2j=n$.
In this case, the LHS of \eqref{BFY} is 
\begin{align}
&\frac{1}{n+2}
\sum_{m=-j}^{j}
\frac{j+m+1}{2j+1}
{n \choose j+m} p^{j+m}(1-p)^{j-m}\notag\\
=&
\frac{1}{n+2}\sum_{k=0}^{n}
\frac{k+1}{n+1}
{n \choose k} p^{k}(1-p)^{n-k}
=
\frac{1}{n+2}\frac{np+1}{n+1}.
\end{align}
The RHS of \eqref{BFY} is 
\begin{align}
&
\frac{1}{n}\sum_{m=-j}^{j}
\frac{j-m}{2j+1}
{n \choose j+m} p^{j+m}(1-p)^{j-m}\notag\\
=&
\frac{1}{n}\sum_{k=0}^{n}
\frac{n-k}{n+1}
{n \choose k} p^{k}(1-p)^{n-k}
=
\frac{1}{n}\frac{n(1-p)}{n+1}.
\end{align}

When $n$ goes to infinity, 
the limit of $n$ times of LHS of \eqref{BFY} equals $p$
and
the limit of $n$ times of RHS of \eqref{BFY} equals $1-p$.
When $p > 1/2$, with sufficiently large $n$, 
the condition \eqref{BFY} holds.

Then, 
the maximum of $R(M)$ is 
$(n+1)\frac{1}{n+2}\frac{np}{n+1}
=\frac{np}{n+2}$, which converges to $p$.
Hence the error function $\eta(M) = 4(1-R(M))$ converges to $4(1-p)$, which is strictly larger than $0$. 
Thus, we cannot make a precise global estimate of $\theta$ even with sufficiently large $n$, and the Cram\'er-Rao approach does not work well for the global estimation problem in this case.

\subsubsection{Probe state as a geometric distribution in the number basis}\label{geo-dis}

When $p_m$ is a geometric distribution
$\frac{r-1}{r^{2j+1}-1}r^{j+m}$,
the LHS of \eqref{BFY} is 
\begin{align}
&\frac{1}{n+2}
\sum_{m=-j}^{j}
\frac{j+m+1}{2j+1}
\frac{r-1}{r^{2j+1}-1}r^{j+m}\notag\\
=&
\frac{1}{n+2}\Big(
\frac{1}{n+1}
+
\frac{1}{n+1}
\big(\frac{n r^{n+1}+1}{r^{n+1}-1}-\frac{1}{r-1}\big)
\Big).
\end{align}
The RHS of \eqref{BFY} is 
\begin{align}
&
\frac{1}{n}\sum_{m=-j}^{j}
\frac{j-m}{2j+1}
\frac{r-1}{r^{2j+1}-1}r^{j+m}\notag\\
=&\frac{1}{n}\sum_{m=-j}^{j}
\frac{2j- (j+m)}{2j+1}
\frac{r-1}{r^{2j+1}-1}r^{j+m}\notag\\
=&\frac{1}{n}
\Big(\frac{n}{n+1}
-\frac{1}{n+1}\big(\frac{n r^{n+1}+1}{r^{n+1}-1}-\frac{1}{r-1}\big)
\Big).
\end{align}

When $r>1$ and 
$n$ goes to infinity, 
LHS of \eqref{BFY} approaches $1$ and 
and
RHS of \eqref{BFY} approaches $0$.
With sufficiently large $n$, the condition \eqref{BFY} holds.
Then, 
the maximum of $R(M)$ is 
$
\frac{n+1}{n+2}\Big(
\frac{1}{n+1}
+
\frac{1}{n+1}
(\frac{n r^{n+1}+1}{r^{n+1}-1}-\frac{1}{r-1})
\Big)
$, which converges to $1$, where
$R(M)$ is defined in \eqref{BNVX}.
When $M_n$ is the optimal estimator,  
we show in Appendix \ref{app:erfgeo} that the corresponding error function is
\begin{align}
\eta(M) = 
\frac{4r}{n(r-1)}
-\frac{8}{n^2(r-1)}
+O(n^{-3})
+O(r^{-n-1}).\label{NM7}
\end{align}
Therefore, 
the minimum error for the global estimate coincides with the RLD bound \eqref{BZV5}.
In this case, the Cram\'er-Rao approach works well for our global estimation problem.

\subsubsection{Probe state as a delta distribution in the number basis}

Next consider the case where $p_m = \delta_{a,m}$. 
Then the LHS of \eqref{BFY} is 
$\frac{1}{n+2} \frac{j+a+1}{2j+1}$.
The RHS of \eqref{BFY} is given by 
$\frac{1}{n}\frac{j-a}{2j+1}.$
The difference between the LHS and the RHS of \eqref{BFY} gives the expression
\begin{align}
&\frac{1}{n(n+2)}
\left( 
n (j+a+1) - (n+2)(j-a)
\right)\notag\\
=&
%\frac{1}{n(n+2)}
%\left( 
%n j+n a+ n 
%- n j + n a
%- 2j + 2a
%\right)\notag\\
%=&
\frac{1}{n(n+2)}
\left( 
2 (n+1) a + (n - 2j) 
\right). \label{Dicke-condition}
\end{align}
Since $n=2j$, the expression in \eqref{Dicke-condition}
tells us that \eqref{BFY} is equivalent to the inequality
\begin{align}
 a  \ge 0  .  
\end{align}
Hence whenever we have a state $|n/2;a\>$ with $a \ge 0$, the condition \eqref{BFY} holds, and the maximum of $R(M)$ is 
$\frac{n+1}{n+2} \frac{n/2+a+1}{n+1}$.
In the limit of large $n$ becomes large, this maximum $R(M)$ 
becomes $\frac{1}{2} + \frac{a}{n}$.
For positive $a$, this $R(M)$ is at least $\frac{1}{2}$, and is bounded away from zero. Hence the global estimation strategy for such states in the basis $\mathcal B_n$ has a constant error.
In contrast, the local estimation strategy has MSE that scales as $O(1/n^2)$.
Hence the Cram\'er-Rao approach does not work well for the global estimation problem in this case.

As an example, we may consider the probe state given by $|n/2;0\>$ which corresponds to using a half-Dicke state for distinguishable spins. 
The state $|n/2;0\>$, 
commonly discussed as a quantum probe state that we can use for a quantum advantage in quantum sensing \cite{PhysRevA.85.022321,tothPhysRevA.85.022322}, can be prepared for instance in the procedure described in Ref~\cite{lucke2011twin,halfdicke}.
According to \eqref{NMI}, $F_{1,1}=F_{2,2} = 8j(j+1) = 2n(n+2)$,
and the tight CR bound scales as $O(1/n^2)$. 
However, for global estimation strategies, the minimum MSE is a constant, because  $R(M)=1/2$.
Hence under global estimation strategies, the half-Dicke state loses its quantum advantage in sensing.

\section{Discussion}\label{S6}

We have shown that that there are situations where for the state estimation problem,
the minimum MSE obtained from the CR approach 
is accurate for the error function obtained for global estimation strategies. 
We have also shown that the opposite can be true;
namely, there are situations where for the state estimation problem,
the minimum MSE obtained from the CR approach 
is very different from the error function obtained for global estimation strategies. 
The most striking difference between the minimum MSE obtained from the CR approach and the minimum error function from global estimation 
is the situation of estimating a unitary model
with the probe state $|n; n/2\>_{B} = |n/2; 0\>$
In the context of local estimation, there is a Heisenberg scaling in the minimum MSE if we use $|n/2; 0\>$ as the probe state for this unitary model.
However in the limit of large $n$, we show that this Heisenberg scaling vanishes for global estimation strategies.
Our results recommends that caution must be exercised if we wish to use CR bounds in the context of global estimation.

In the case of unitary estimation of a single parameter, it is known that
the optimal Cram\'{e}r-Rao type bound cannot be attained with global estimation \cite{Ha06,Ha11-2,MH16-9}. 
The papers \cite{Ha06,Ha11-2,MH16-9} consider the optimization of the initial state for the estimation of the unitary.
However, this paper considers the state estimation with a fixed
initial state.
Furthermore, the paper \cite{Ha11-2} showed that
the optimal Cram\'{e}r-Rao type bound cannot be attained 
even under the problem of local minimax estimation even under the 
setting 
of asymptotically many probe states used for global estimation of unitary channels.
This phenomenon relates to the Heisenberg scaling under the unitary estimation.
In unitary estimation, 
while the optimal Cram\'{e}r-Rao type bound has the same order as
the error function for optimum minimax estimation \cite{F02,IF07,Ha06,Ha11-2,MH16-9},
they differ in the coefficients of their leading order terms.
Moreover, it was shown that the input state of 
the optimal Cram\'{e}r-Rao type bound and the optimum minimax estimation
are different.
For example, although the noon state realizes
the optimal Cram\'{e}r-Rao type bound, it does not work for global estimation \cite[Section VI]{mh-photon-constant}.
Our results add to the literature of examples where 
the behavior of Cram\'{e}r-Rao type bounds differs from
the optimum minimax estimation, particularly with regards to the quantum estimation of bosonic states both in a single-parameter and a multi-parameter setting.

Our work is also related to the question as to 
whether a family of bosonic states which embed parameters can have a
Cram\'{e}r-Rao type bound that can be attained in the single copy setting.
One family of quantum states that we considered is 
the state family which is a geometric distribution in the Fock basis.
In Ref.~\cite[Section IV]{HM08} showed that this state family approximates a
quantum Gaussian state family.
Moreover, in the setting of multiple identical and independently distributed copies,
this geometric state family converges to the quantum Gaussian state family \cite{KG,Kahn,two4}.
Given that the RLD bound can be attained under the single copy setting
for the quantum Gaussian state family,
we can see that our problem relates to the question as to 
whether
the state family of our interest also attains the Cram\'{e}r-Rao type bound in the single copy setting.
We leave this line of enquiry for future work. 

%Recently, the discrepancy between the quantum Fisher information and the classical Fisher information was discussed on a two-mode bosonic system with certain physical constraints was observed \cite{uwePhysRevLett.132.240803}. 
%This leaves open the question as to what the optimal MSE would be in both in the local and global estimation setting for this type of physical system with constrained measurements.

Quantum state estimation for two interacting modes in a double well system has been studied, and in the context of local estimation theory, measurements using the widely considered two-mode interferometer was shown to be strictly suboptimal \cite{uwePhysRevLett.132.240803}. 
This is an interesting result because although the two-mode interferometer is often optimal for the local estimation, the two-mode interferometer with fixed orbitals is not optimal in the situation of Ref. \cite{uwePhysRevLett.132.240803} using the usual two-mode interferometer with fixed orbitals; two-mode interferometry in this case cannot capture the Fisher information correctly at all, not even approximately.
For future work, it would be interesting to consider the question of whether the two-mode interferometer with fixed orbitals is still suboptimal in the global estimation setting.

We also like to comment on the relationship between the quantum state estimation problems that we have studied with the research direction of permutation-invariant quantum tomography \cite{pitomograph-PhysRevLett.105.250403,moroder2012permutationally,PhysRevLett.113.040503,PhysRevLett.113.040503}.
These papers discuss state tomography over the symmetric subspace;
namely linear tomography on only the symmetric subspace by using Dicke states as orthogonal basis is considered.
In contrast, our paper considers the state estimation under the assumption that the true state belongs to the unitary orbit of a Dicke state.
Additionally, our paper allows more types of measurements, namely any covariant measurement (which allows us to employs the group covariance method). 
In particular, our theoretical analysis shows our method to be optimal in the global estimation setting.
Since our method allows more types of measurements, our estimator has better performance than one that projects only projects onto Dicke states.

Next, we discuss the global attainability of the Cramér-Rao bound, which represents the minimum of the mean squared error (MSE) sum under the local unbiasedness condition.
In the single-parameter case, 
the Cramér-Rao bound coincides with
the optimal error under the local unbiasedness condition
if and only if the state family is an exponential family in the sense of SLD (symmetric logarithmic derivative) \cite{nagaoka89selectedpapers,nagaoka87selectedpapers}\cite[Theorem 6.7]{hayashi2016quantum}. 
In the multiple-parameter case, 
it is difficult to obtain the same assertion due to the complexity of problems associated with non-commutativity in minimizing the MSE sum under the local unbiasedness condition. 
However, its connection with exponential families has been discussed in recent literature \cite{Nagaoka23}. At least in the classical case, under certain regularity conditions on the distribution family, being an exponential family is equivalent to the Cramér-Rao bound being globally attainable.

It is conceivable that as the state family approaches an exponential family, the Cramér-Rao bound becomes globally attainable. The quantum Gaussian states family is shown to be a non-commutative exponential family \cite[Section 6]{Nagaoka23}. 
For instance, in state estimation, the Cramér-Rao bound is globally attainable when the number of copies is large. 
As shown in \cite{KG,Kahn,two4}, in this case, it asymptotically approaches the quantum Gaussian states family, albeit locally. Furthermore, the state family treated in Subsection \ref{geo-dis} also asymptotically approaches the quantum Gaussian states family, as shown in \cite[Theorem 7]{HM08} in the limit of large $j$. 
Therefore, as confirmed in \eqref{NM7}, the Cramér-Rao bound is asymptotically globally attainable. 
In other words, it is safer to assume that the Cramér-Rao bound is not globally attainable unless the situation approaches an exponential family.

Finally, we discuss the relation between the unbiased condition and the group covariant approach discussed in Section \ref{Global-estimation}. Since our state space forms a compact set, 
any coordinate cannot cover the whole of the parameter space smoothly. 
Now, we employ the parameter space 
$\{(\theta_1,\theta_2)| |\theta| \le \pi \}$,
and this space has discontinuity at the state 
$|\frac{1}{2};-\frac{1}{2}\rangle \langle \frac{1}{2};-\frac{1}{2}|$.
Due to this issue, any covariant estimator including our optimal estimator has bias except for the origin $\theta=0$ under this parametrization because the center of the measurement outcome is shifted to the direction of the origin.
Since the Cramér-Rao approach employs only the locally unbiased condition and Section \ref{local:estimation:S4} discusses the error at the origin, it is more important whether 
our optimal estimator satisfies the locally unbiased condition. 
However, it is difficult to check whether 
a covariant estimator satisfies the locally unbiased condition
because the quantities appearing in this condition 
cannot be easily handled in the group covariant framework.
Nevertheless, the state family treated in Subsection \ref{geo-dis} also asymptotically approaches the quantum Gaussian states family 
in the limit of large $j$, as mentioned the above.
Since the optimal covariant estimator for the quantum Gaussian states family satisfies the unbiasedness condition \cite[Chapter 6]{Holevo}, 
the optimal covariant estimator asymptotically
satisfies the unbiasedness condition in this case.

\section*{Acknowledgement}
MH is supported in part by the National Natural Science Foundation of China (Grant No. 62171212),
and 
the General R\&D Projects of 
1+1+1 CUHK-CUHK(SZ)-GDST Joint Collaboration Fund 
(Grant No. GRDP2025-022).
YO acknowledges support from EPSRC (Grant No. EP/W028115/1).

\appendix

\section{Derivation of \eqref{fidelity2}} \label{AA1}
We have
\begin{align}
&U_{\theta} := \exp(i (\theta_1 J_1+\theta_2 J_2)) %\notag\\
=\exp \left(
\begin{pmatrix}
0 & \frac{i \theta_1+\theta_2}{2} \notag\\
\frac{i \theta_1-\theta_2}{2} & 0 
\end{pmatrix}
\right) \notag\\
%=&\exp \left(\left( \begin{array}{cc}
%0 & \frac{i( \theta_1-i\theta_2)}{2} \\
%\frac{i( \theta_1+i\theta_2)}{2} & 0 
%\end{array}\right)\right) \notag\\
=&\exp \left(
\left(
\begin{array}{cc}
0 & ie^{-i\phi} \frac{|\theta|}{2}  \\
ie^{i\phi} \frac{|\theta|}{2} & 0 
\end{array}
\right)
\right) \notag\\
=&\exp \left(
\left(
\begin{array}{cc}
1 &   0\\
0& -ie^{i\phi} 
\end{array}
\right)
\left(
\begin{array}{cc}
0 &  \frac{|\theta|}{2}  \\
- \frac{|\theta|}{2} & 0 
\end{array}
\right)
\left(
\begin{array}{cc}
1 &  0 \\
0& ie^{-i\phi} 
\end{array}
\right)
\right) \notag\\
=&
\left(
\begin{array}{cc}
1 &   0\\
0& -ie^{i\phi} 
\end{array}
\right)
\exp \left(
\left(
\begin{array}{cc}
0 &  |\theta|  \\
- |\theta| & 0 
\end{array}
\right)
\right)
\left(
\begin{array}{cc}
1 &  0 \\
0& ie^{-i\phi} 
\end{array}
\right) \notag\\
=&
\left(
\begin{array}{cc}
1 &   0\\
0& -ie^{i\phi} 
\end{array}
\right)
\left(
\begin{array}{cc}
\cos \frac{|\theta|}{2} & \sin \frac{|\theta|}{2}  \\
-\sin \frac{|\theta|}{2} & \cos \frac{|\theta|}{2}
\end{array}
\right)
\left(
\begin{array}{cc}
1 &  0 \\
0& ie^{-i\phi} 
\end{array}
\right) \notag\\
=&
\left(
\begin{array}{cc}
\cos \frac{|\theta|}{2} & ie^{-i\phi}  \sin \frac{|\theta|}{2}  \\
ie^{i\phi} \sin \frac{|\theta|}{2} & \cos \frac{|\theta|}{2}
\end{array}
\right).
\end{align}
Thus,
\begin{align}
&U_{\theta} |\frac{1}{2};\frac{1}{2}\rangle
=\cos \frac{|\theta|}{2}|\frac{1}{2};\frac{1}{2}\rangle
+ie^{i\phi}\sin \frac{|\theta|}{2}  |\frac{1}{2};-\frac{1}{2}\rangle .
\end{align}
Hence, 
\begin{align}
R(\theta,\hat{\theta})=&|\langle \frac{1}{2};\frac{1}{2}| U_{\theta}^\dagger U_{\hat{\theta}} |
\frac{1}{2};\frac{1}{2}\rangle|^2 \notag\\
=&\Big|
\cos \frac{|\theta|}{2} \cos \frac{| \hat{\theta}|}{2}
+e^{i(-\phi +\hat{\phi})} \sin \frac{|\theta|}{2}\sin 
\frac{|\hat{\theta}|}{2}
\Big|^2.
\end{align}

\section{Local estimation under the general unitary model}\label{A1}
To show several relations in Section \ref{S4},
we consider the spin $j$ system ${\cal H}_j$ spanned by 
$\{|j;m\rangle\}_{m=-j}^j$.
For simplicity, if $j$ is clear from the context, we denote 
$|j;m\rangle$ as $|m\rangle$.

We have two operators as
\begin{align}
J_+:
=&\sum_{m=-j}^{j-1}  \sqrt{j(j+1)-m(m+1)}|m+1\rangle \langle m|\notag\\
=&\sum_{m=-j}^{j-1}  \sqrt{(j-m)(j+m+1)}|m+1\rangle \langle m|\\
J_-:=&\sum_{m=-j}^{j-1}  \sqrt{j(j+1)-m(m+1)}|m\rangle \langle m+1|\notag\\
=&\sum_{m=-j}^{j-1}  \sqrt{(j-m)(j+m+1)}|m\rangle \langle m+1|.
\end{align}
We write the angular momentum operators as
\begin{align}
J_1:=&\frac{1}{2}(J_++J_-), ~
J_2:=\frac{1}{2i }(J_+-J_-) \\
J_3:= &\sum_{m=-j}^{j}  m|m\rangle \langle m|.
\end{align}
Then, the Casimir element $C$ is given as
\begin{align}
C:=&\sum_{k=1}^3 J_k^2
=\frac{1}{4}((J_++J_-)^2-(J_+-J_-)^2)+J_3^2 \notag\\
=&\frac{1}{2}(J_+ J_- +J_-J_+)+J_3^2
= j(j+1).
\end{align}

Now, we consider the following case
\begin{align}
\rho&:=\sum_{m=-j}^{j}  p_m|m\rangle \langle m| \\
D_1&:=i [\rho,J_1],~
D_2:=i [\rho,J_2].
\end{align}
Then, we have
\begin{align}
i [\rho, J_1]&= \frac{i}{2}[\rho, (J_+ + J_-)]
\notag\\
&= 
\frac{i}{2} (-J_+\rho+ \rho J_+ -J_-\rho +\rho J_- )\\
i [\rho, J_2]&= \frac{1}{2}[\rho, (J_+ - J_-)]
\notag\\
&= 
\frac{1}{2} (-J_+\rho+ \rho J_+ +J_-\rho -\rho J_- ).
\end{align}

We define
\begin{align}
K_+:= &
\sum_{m=-j}^{j-1}  \frac{2(p_{m+1}-p_m)}{p_{m+1}+p_m }
\notag\\
&\quad \quad \times 
\sqrt{(j-m)(j+m+1)}|m+1\rangle \langle m|\\
K_-:= &
\sum_{m=-j}^{j-1}  \frac{2(p_{m+1}-p_m)}{p_{m+1}+p_m }
\notag\\
&\quad \quad \times 
\sqrt{(j-m)(j+m+1)}|m\rangle \langle m+1|.
\end{align}
Then, we have
\begin{align}
\rho \circ (\frac{1}{2}(K_++K_-))=& i [\rho, J_2] =D_2\label{N1}\\
\rho \circ (-\frac{1}{2i }(K_+-K_-))= &i [\rho, J_1] =D_1.\label{N2}
\end{align}
Thus, 
the SLDs $L_1$ and $L_2$ of the first and second parameters
are calculated as 
\begin{align}
{L}_2:=&\frac{1}{2}(K_++K_-), ~
{L}_1:=-\frac{1}{2i }(K_+-K_-) .
\end{align}
Thus, we have
\begin{align}
&F_{1,1}= \tr \rho L_1^2 \notag\\
=&
\sum_{m=-j}^{j-1}  \Big(
\frac{4(p_{m+1}-p_m)^2}{(p_{m+1}+p_m)^2 }
(j-m)(j+m+1) p_m\notag\\
&+
\frac{4(p_{m+1}-p_m)^2}{(p_{m+1}+p_m)^2 }
(j-m)(j+m+1) p_{m+1} \Big)\notag\\
=&
\sum_{m=-j}^{j-1}  
\frac{4(p_{m+1}-p_m)^2}{p_{m+1}+p_m }
(j-m)(j+m+1) ,\\
&F_{2,2}= \tr \rho L_2^2 \notag\\
=&
\sum_{m=-j}^{j-1}  
\frac{4(p_{m+1}-p_m)^2}{p_{m+1}+p_m }
(j-m)(j+m+1) ,\\
&F_{1,2}= \tr \rho \tilde{L}_1 \circ \tilde{L}_2=0,
\end{align}
which shows \eqref{NMI}.

In particular, when the distribution $\{p_m\}$ is a geometric distribution
$p_{G,m}=
\frac{r-1}{r^{j+1}-r^{-j}}r^{m}$,
we have
\begin{align}
\frac{2(p_{m+1}-p_m)}{p_{m+1}+p_m }
=
\frac{2(r^{m+1}-r^{m})}{r^{m+1}+r^{m} }
=
\frac{2(r-1)}{r+1},
\end{align}
which implies that
\begin{align}
K_+=\frac{2(r-1)}{r+1} J_+,~
K_-=\frac{2(r-1)}{r+1} J_-.
\end{align}
Using \eqref{N1} and \eqref{N2}, we have
\begin{align}
\rho \circ L_2= i [\rho, \frac{r+1}{2(r-1)}L_1], ~
\rho \circ L_1= i [\rho, \frac{r+1}{2(r-1)}L_2] .
\end{align}
These relations guarantee the D-invariance.

\begin{widetext}
\section{Local estimation with the binomial distribution: Proof of \eqref{BZV}}

We recall the Clebsch–Gordan formula, which for $2j=n$ gives
\begin{align}
&\langle j+\frac{1}{2}; m+\frac{1}{2}
|\frac{1}{2},\frac{1}{2};j,m\rangle^2  
=\frac{(2j+2)(2j)! (j+m)!(j-m)!(j+m+1)!(j-m)!}
{(2j+2)! ((j+m)!  (j-m)!)^2}\notag\\
=&
\frac{(2j+2) (j+m+1)}
{(2j+2)(2j+1)}
=
\frac{(j+m+1)}{(2j+1)}.
\end{align}
Now we also consider 
$p_m= {n \choose j+m} p^{j+m}(1-p)^{j-m}$.

Then, we have
\begin{align}
F_{1,1}
= &
\sum_{m=-j}^{j-1}  
\frac{4(p_{m+1}-p_m)^2}{p_{m+1}+p_m }
(j-m)(j+m+1) \notag \\
= &
\sum_{m=-j}^{j-1}  
\frac{4({n \choose j+m+1} p^{j+m+1}(1-p)^{j-m-1}
-{n \choose j+m} p^{j+m}(1-p)^{j-m})^2}{
{n \choose j+m+1} p^{j+m+1}(1-p)^{j-m-1}
+{n \choose j+m} p^{j+m}(1-p)^{j-m}}\notag\\
&\cdot(j-m)(j+m+1) \notag\\
= &
\sum_{m=-j}^{j-1}  
\frac{4
(\frac{j-m}{j+m+1}\frac{p}{1-p}-1)^2
{n \choose j+m} p^{j+m}(1-p)^{j-m}}
{\frac{j-m}{j+m+1}\frac{p}{1-p}+1}
(j-m)(j+m+1) \notag\\
= &
\sum_{m=-j}^{j-1}  
\frac{4
((j-m)\frac{p}{1-p}+(j+m+1))^2
{n \choose j+m} p^{j+m}(1-p)^{j-m}}
{(j-m)\frac{p}{1-p}+(j+m+1)}
(j-m)\notag\\
= &
\sum_{m=-j}^{j-1}  
\frac{4
((j-m)p-(j+m+1)(1-p))^2
{n \choose j+m} p^{j+m}(1-p)^{j-m-1}}
{(j-m)p+(j+m+1)(1-p)}
(j-m)\notag\\
= &
\sum_{m=-j}^{j-1}  
\frac{4
(-(j+m+1) +(2j+1)p)^2
{n \choose j+m} p^{j+m}(1-p)^{j-m-1}}
{j +m+1 -(2m+1)p}
(j-m)\notag\\
= &
\sum_{k=0}^{n-1}  
\frac{4
(-(k+1) +(n+1)p)^2
{n \choose k} p^{k}(1-p)^{2n-k-1}}
{k+1 -(2k-n+1)p}
(n-k)\notag\\
= &
n^2 \sum_{k=0}^{n-1}  
\frac{4
(-\frac{k+1}{n} +(1+\frac{1}{n})p)^2
{n \choose k} p^{k}(1-p)^{2n-k-1}}
{ \frac{k+1}{n} -(\frac{2k}{n}-1+\frac{1}{n})p}
(1-\frac{k}{n}).
\end{align}
Now we interpret the index $k$ as a measurement outcome we obtain from measuring the state in the Fock basis, and $Y=k/n$ as the corresponding random variable. Then we can write $F_{1,1}$ in terms of $Y$ to get
\begin{align}
F_{1,1}
=&
n^2 
\mathbb{E}_{p, \frac{p(1-p)}{n} }
\Big[
\frac{4
(-Y-\frac{1}{n} +(1+\frac{1}{n})p)^2
(1-p)^{-1}(1-Y)}
{ Y+\frac{1}{n} -(2Y -1+\frac{1}{n})p}
\Big],
\end{align}
where the first subscript on the expectation denotes the mean of $Y$, and the second subscript denotes the variance of $Y$.
We then define the random variable $X=\sqrt{n}(Y-p)$, and get
\begin{align}
F_{1,1}
=&n^2
\mathbb{E}_{0, p(1-p)}
\Big[
\frac{4
(-p-\frac{X}{\sqrt{n}} +(1+\frac{1}{n})p)^2
(1-p)^{-1}(1-p-\frac{X}{\sqrt{n}})}
{ p+\frac{X}{\sqrt{n}}+\frac{1}{n} -(
2p+2\frac{X}{\sqrt{n}}-1+\frac{1}{n})p}
\Big] \notag\\
= & 
n^2
\mathbb{E}_{0, p(1-p)}
\Big[
\frac{4
(\frac{X}{\sqrt{n}} -\frac{p}{n})^2
(1-p)^{-1}(1-p-\frac{X}{\sqrt{n}})}
{-2p^2+2p +(1-2p) \frac{X}{\sqrt{n}}+\frac{1-p}{n}
}
\Big]\notag\\
\cong& 
n^2
\mathbb{E}_{0, p(1-p)}
\Big[
\frac{2 \frac{X^2}{n}
(1-p)^{-1}(1-p)}
{p(1-p)}
\Big]
=
n
\mathbb{E}_{0, p(1-p)}
\Big[
\frac{2 X^2
}
{p(1-p)}
\Big]=2n,
\end{align}
where the congruent symbol indicates an approximation in the limit of large $n$.
Hence, we obtain \eqref{BZV}.

\section{Local estimation with the geometric distribution}\label{A-D}
Assume that 
$p_m= \frac{r-1}{r^{2j+1}-1}r^{j+m}$.
Since 
\be
\frac{2(r^{j+m+1}-r^{j+m})}{r^{j+m+1}+r^{j+m}}=
\frac{2(r-1)}{r+1},
\ee
we have
\begin{align}
K_+= &
\sum_{m=-j}^{j-1} \frac{2(r-1)}{r+1}
\sqrt{(j-m)(j+m+1)}|m+1\rangle \langle m|
=\frac{2(r-1)}{r+1}J_+
\\
K_-= &
\sum_{m=-j}^{j-1} \frac{2(r-1)}{r+1}
\sqrt{(j-m)(j+m+1)}|m\rangle \langle m+1|
=\frac{2(r-1)}{r+1}J_-.
\end{align}
Thus, we have
\begin{align}
{L}_2=&\frac{2(r-1)}{r+1}J_1, ~
{L}_1=-\frac{2(r-1)}{r+1}J_2.
\end{align}

Thus, we have
\begin{align}
&F_{1,1}= F_{2,2}=
\sum_{m=-j}^{j-1}  
\frac{4(p_{m+1}-p_m)^2}{p_{m+1}+p_m }
(j-m)(j+m+1) \notag\\
=&
\sum_{m=-j}^{j-1}  
4 p_m (r-1)\frac{2(r-1)}{r+1}
(j-m)(j+m+1) 
\cong 
\frac{n}{2} \frac{2(r-1)}{r+1}.
\end{align}

Relation \eqref{N1} and \eqref{N2} are rewritten as
\begin{align}
-\rho \circ \frac{2(r+1)}{r-1}J_2= i [\rho, J_1] ,\quad
\rho \circ \frac{2(r+1)}{r-1}J_1= i [\rho, J_2].
\end{align}
Thus, this model satisfies the D-invariant condition.
The matrix $D=(D_{j,k})$ is approximated to
\begin{align}
\frac{n}{2} \frac{2(r-1)}{r+1}
\cdot
\frac{2(r-1)}{r+1}
\left(
\begin{array}{cc}
0 & 1 \\
-1 & 0
\end{array}
\right).
\end{align}
Thus
\begin{align}
C^{\rm R}_{\bm{\theta}}[I]=\Tr F^{-1}+\frac{1}{2}\Tr |F^{-1}D F^{-1} | %\notag\\
\cong
\frac{2(r+1)}{n(r-1)} +\frac{2}{n}
=\frac{2(r+1)+2(r-1)}{n(r-1)}
=\frac{4r}{n(r-1)},
\end{align}
which implies \eqref{BZV5}.

\section{Proof of Theorem \ref{TH2}}\label{A-F}
For reader's convenience, 
we give a proof for Theorem \ref{TH2} 
as a special case of \cite[Theorem 4.6]{Group2}.

When our estimator $\hat{\theta}$ corresponds to $g \in \SU(2)$ 
and when the true parameter is 0,
the fidelity function as given in \eqref{fidelity} can be expressed as 
\begin{align}
R(0,\hat{\theta})=
\cos^2 \frac{|\hat{\theta}|}{2}=
\Tr |\frac{1}{2};\frac{1}{2}\rangle \langle \frac{1}{2};\frac{1}{2}|
U_{1/2,g} 
|\frac{1}{2};\frac{1}{2}\rangle \langle \frac{1}{2};\frac{1}{2}|
U_{1/2,g}^\dagger 
\end{align}
where $U_{1/2,g}$ is a spin-1/2 unitary representation of $g \in \SU(2)$. 

Now, let us write the true state as $\rho$,
and write the POVM element that corresponds to $g \in \SU(2)$ as 
$(2j+1) U_{j,g} \rho' U_{j,g}^\dagger $ for some state $\rho'$. 
We consider the case when our measurement is given as $M_{\rho'}$.
Then, by using the Haar measure $\nu$ on $\SU(2)$,
the Bayesian average of $R(0,\hat{\theta})$ is calculated as
\begin{align}
&R(M_{\rho'})=
\int_{\SU(2)} 
\Tr |\frac{1}{2};\frac{1}{2}\rangle \langle \frac{1}{2};\frac{1}{2}|
U_{1/2,g} 
|\frac{1}{2};\frac{1}{2}\rangle \langle \frac{1}{2};\frac{1}{2}|
U_{1/2,g}^\dagger 
\cdot (2j+1)\Tr \rho
U_{j,g} 
\rho'
U_{j,g}^\dagger \nu(dg)\notag\\
=&
\int_{\SU(2)} 
(2j+1)\Tr 
|\frac{1}{2};\frac{1}{2}\rangle \langle \frac{1}{2};\frac{1}{2}|\otimes \rho
(U_{1/2,g} \otimes U_{j,g} )
|\frac{1}{2};\frac{1}{2}\rangle \langle \frac{1}{2};\frac{1}{2}|
\otimes \rho'
(U_{1/2,g} \otimes U_{j,g} )^\dagger \nu(dg)
\notag\\
=&
(2j+1)\Tr 
(P_{j+\frac{1}{2}}+P_{j-\frac{1}{2}})
(|\frac{1}{2};\frac{1}{2}\rangle \langle \frac{1}{2};\frac{1}{2}|\otimes \rho)
(P_{j+\frac{1}{2}}+P_{j-\frac{1}{2}})
\notag\\
&\quad \times
\int_{\SU(2)} 
(U_{1/2,g} \otimes U_{j,g} )
(|\frac{1}{2};\frac{1}{2}\rangle \langle \frac{1}{2};\frac{1}{2}|
\otimes \rho')
(U_{1/2,g} \otimes U_{j,g} )^\dagger \nu(dg).
\end{align}
Here, $P_j$ is a projector onto the space with total spin of $j$.
We obtain the above by using the properties of the trace function, which allows us to rewrite the Bayesian average as the integral of the trace of operators on a tensor product space.
In the next step we use properties of how a spin-$j$ space combines with a spin-1/2 space.
\begin{align}
R(M_{\rho'})
\stackrel{(a)}{=}&
(2j+1)\Tr 
P_{j+\frac{1}{2}}
(|\frac{1}{2};\frac{1}{2}\rangle \langle \frac{1}{2};\frac{1}{2}|\otimes \rho)
P_{j+\frac{1}{2}}
\int_{\SU(2)} 
(U_{1/2,g} \otimes U_{j,g} )
(|\frac{1}{2};\frac{1}{2}\rangle \langle \frac{1}{2};\frac{1}{2}|
\otimes \rho')
(U_{1/2,g} \otimes U_{j,g} )^\dagger \nu(dg) \notag
\\
&+
(2j+1)\Tr 
P_{j-\frac{1}{2}}
(|\frac{1}{2};\frac{1}{2}\rangle \langle \frac{1}{2};\frac{1}{2}|\otimes \rho)
P_{j-\frac{1}{2}}
\int_{\SU(2)} 
(U_{1/2,g} \otimes U_{j,g} )
(|\frac{1}{2};\frac{1}{2}\rangle \langle \frac{1}{2};\frac{1}{2}|
\otimes \rho')
(U_{1/2,g} \otimes U_{j,g} )^\dagger \nu(dg)
\end{align}
To obtain this note that since the operator 
$\int_{\SU(2)} 
(U_{1/2,g} \otimes U_{j,g} )
(|\frac{1}{2};\frac{1}{2}\rangle \langle \frac{1}{2};\frac{1}{2}|
\otimes \rho')
(U_{1/2,g} \otimes U_{j,g} )^\dagger \nu(dg)$
is commutative with $(U_{1/2,g'} \otimes U_{j,g'})$,
Schur's lemma guarantees that this operator is a sum of constant times of 
$P_{j+\frac{1}{2}}$ and $P_{j-\frac{1}{2}}$.
Hence, this operator is commutative with 
$P_{j+\frac{1}{2}}$ and $P_{j-\frac{1}{2}}$.
Hence, we obtain Step $(a)$.

Next, we proceed to evaluate the integrals and obtain
\begin{align}
R(M_{\rho'})\stackrel{(b)}{=}&
(2j+1)\Tr 
P_{j+\frac{1}{2}}
(|\frac{1}{2};\frac{1}{2}\rangle \langle \frac{1}{2};\frac{1}{2}|\otimes \rho)
\frac{1}{2j+2}
(\Tr P_{j+\frac{1}{2}} |\frac{1}{2};\frac{1}{2}\rangle \langle \frac{1}{2};\frac{1}{2}|\otimes \rho')
P_{j+\frac{1}{2}}\notag
\\
&+
(2j+1)\Tr 
P_{j-\frac{1}{2}}
(|\frac{1}{2};\frac{1}{2}\rangle \langle \frac{1}{2};\frac{1}{2}|\otimes \rho)
\frac{1}{2j}
(\Tr P_{j-\frac{1}{2}} |\frac{1}{2};\frac{1}{2}\rangle \langle \frac{1}{2};\frac{1}{2}|\otimes \rho')
P_{j-\frac{1}{2}}
\end{align}
Since the irreducibility of the spaces 
${\cal H}_{j+\frac{1}{2}}$ and
${\cal H}_{j-\frac{1}{2}}$ guarantees that 
the first and second integral terms in Step (a) equal 
constant times of 
$P_{j+\frac{1}{2}}$ and $P_{j-\frac{1}{2}}$, respectively,
we obtain Step $(b)$.

Finally, we simplify further and get
\begin{align}
R(M_{\rho'})=&
\frac{2j+1}{2j+2}
\Tr 
P_{j+\frac{1}{2}}
(|\frac{1}{2};\frac{1}{2}\rangle \langle \frac{1}{2};\frac{1}{2}|\otimes \rho)
(\Tr P_{j+\frac{1}{2}} |\frac{1}{2};\frac{1}{2}\rangle \langle \frac{1}{2};\frac{1}{2}|\otimes \rho')
\notag\\
&+
\frac{2j+1}{2j}
\Tr 
P_{j-\frac{1}{2}}
(|\frac{1}{2};\frac{1}{2}\rangle \langle \frac{1}{2};\frac{1}{2}|\otimes \rho)
(\Tr P_{j-\frac{1}{2}} |\frac{1}{2};\frac{1}{2}\rangle \langle \frac{1}{2};\frac{1}{2}|\otimes \rho')
\notag\\
\stackrel{(c)}{\le}&
\frac{2j+1}{2j+2} \Tr 
P_{j+\frac{1}{2}}
(|\frac{1}{2};\frac{1}{2}\rangle \langle \frac{1}{2};\frac{1}{2}|\otimes \rho).
\end{align}
The inequality $(c)$ follows from 
\eqref{BFY} and the relation 
$(\Tr P_{j+\frac{1}{2}} |\frac{1}{2};\frac{1}{2}\rangle \langle \frac{1}{2};\frac{1}{2}|\otimes \rho')
+
(\Tr P_{j-\frac{1}{2}} |\frac{1}{2};\frac{1}{2}\rangle \langle \frac{1}{2};\frac{1}{2}|\otimes \rho')=1$.
Further, the equality holds when $\rho'=|j;j\rangle \langle j;j|$
because 
$(\Tr P_{j+\frac{1}{2}} |\frac{1}{2};\frac{1}{2}\rangle \langle \frac{1}{2};\frac{1}{2}|\otimes |j;j\rangle \langle j;j|)
=1$.

\section{The error function for the global estimate on the unitary model with the geometric distribution}
\label{app:erfgeo}
With $M_n$ as the optimal covariant estimator for the unitary model on the geometric distribution, we have from Theorem \ref{TH2} that 
\begin{align}
&1-R(M_n)=
1-\frac{n+1}{n+2}\Big(
\frac{1}{n+1}
+
\frac{1}{n+1}
(\frac{n r^{n+1}+1}{r^{n+1}-1}-\frac{1}{r-1})
\Big)
=
1-\Big(
\frac{1}{n+2}
+
\frac{1}{n+2}
(\frac{n r^{n+1}+1}{r^{n+1}-1}-\frac{1}{r-1})
\Big)\notag\\
=&
1-\Big(
\frac{1}{n+2}
+
\frac{1}{n+2}
(\frac{n +r^{-n-1}}{1-r^{-n-1}}-\frac{1}{r-1})
\Big)%\notag\\
= 
1-\Big(
\frac{1}{n+2}
+
\frac{1}{n+2}
( n +O(r^{-n-1})-\frac{1}{r-1})
\Big)\notag\\
= &
\frac{1}{n+2}
\Big(n+2-1- (n +O(r^{-n-1})-\frac{1}{r-1})
\Big)%\notag\\
= 
\frac{1}{n+2}
\Big(1-  \frac{1}{r-1}+O(r^{-n-1})\Big)\notag\\
= &
\frac{r}{(n+2)(r-1)}+O(r^{-n-1}) % \notag\\
= 
\frac{r}{n(r-1)}
-\frac{2}{n^2(r-1)}
+O(\frac{1}{n^3})
+O(r^{-n-1}).
\end{align}
\end{widetext}

\bibliographystyle{quantum}
\bibliography{ref2}

\begin{thebibliography}{10}

\bibitem{Helstrom}
Carl~Wilhelm Helstrom.
\newblock ``Quantum detection and estimation theory''.
\newblock \href{https://dx.doi.org/10.1007/BF01007479}{Academic press}.
  ~(1976).

\bibitem{Holevo}
Alexander~S Holevo.
\newblock ``Probabilistic and statistical aspects of quantum theory''.
\newblock \href{https://dx.doi.org/10.1007/978-88-7642-378-9}{Edizioni della
  Normale}. ~(2011).

\bibitem{two2}
Richard~D. Gill and Serge Massar.
\newblock ``State estimation for large ensembles''.
\newblock \href{https://dx.doi.org/10.1103/PhysRevA.61.042312}{Phys. Rev. A
  {\bf 61}, 042312}~(2000).

\bibitem{two4}
Yuxiang Yang, Giulio Chiribella, and Masahito Hayashi.
\newblock ``Attaining the ultimate precision limit in quantum state
  estimation''.
\newblock \href{https://dx.doi.org/10.1007/s00220-019-03433-4}{Communications
  in Mathematical Physics {\bf 368}, 223--293}~(2019).

\bibitem{HM08}
Masahito Hayashi and Keiji Matsumoto.
\newblock ``Asymptotic performance of optimal state estimation in qubit
  system''.
\newblock \href{https://dx.doi.org/10.1063/1.2988130}{Journal of Mathematical
  Physics {\bf 49}, 102101}~(2008).

\bibitem{HO}
Masahito Hayashi and Yingkai Ouyang.
\newblock ``Tight {C}ram{\'{e}}r-{R}ao type bounds for multiparameter quantum
  metrology through conic programming''.
\newblock \href{https://dx.doi.org/10.22331/q-2023-08-29-1094}{{Quantum} {\bf
  7}, 1094}~(2023).

\bibitem{gorecki2020pi}
Wojciech G{\'o}recki, Rafa{\l} Demkowicz-Dobrza{\'n}ski, Howard~M Wiseman, and
  Dominic~W Berry.
\newblock ``$\pi$-corrected heisenberg limit''.
\newblock \href{https://dx.doi.org/10.1103/PhysRevLett.124.030501}{Physical
  review letters {\bf 124}, 030501}~(2020).

\bibitem{meyer2023quantum}
Johannes~Jakob Meyer, Sumeet Khatri, Daniel~Stilck Fran{\c{c}}a, Jens Eisert,
  and Philippe Faist.
\newblock ``Quantum metrology in the finite-sample regime''~(2023).
\newblock  \href{http://arxiv.org/abs/2307.06370}{arXiv:2307.06370}.

\bibitem{Wojciech}
Wojciech G\'orecki and Rafa\l{} Demkowicz-Dobrza\ifmmode~\acute{n}\else
  \'{n}\fi{}ski.
\newblock ``Multiparameter quantum metrology in the heisenberg limit regime:
  Many-repetition scenario versus full optimization''.
\newblock \href{https://dx.doi.org/10.1103/PhysRevA.106.012424}{Phys. Rev. A
  {\bf 106}, 012424}~(2022).

\bibitem{Chesi}
Giovanni Chesi, Alberto Riccardi, Roberto Rubboli, Lorenzo Maccone, and Chiara
  Macchiavello.
\newblock ``Protocol for global multiphase estimation''.
\newblock \href{https://dx.doi.org/10.1103/PhysRevA.108.012613}{Phys. Rev. A
  {\bf 108}, 012613}~(2023).

\bibitem{Hall_2012}
Michael J~W Hall and Howard~M Wiseman.
\newblock ``Heisenberg-style bounds for arbitrary estimates of shift parameters
  including prior information''.
\newblock \href{https://dx.doi.org/10.1088/1367-2630/14/3/033040}{New Journal
  of Physics {\bf 14}, 033040}~(2012).

\bibitem{Ha11-2}
Masahito Hayashi.
\newblock ``Comparison between the {C}ramer-{R}ao and the mini-max approaches
  in quantum channel estimation''.
\newblock \href{https://dx.doi.org/10.1007/s00220-011-1239-4}{Commun. Math.
  Phys. {\bf 304}, 689 -- 709}~(2011).

\bibitem{Ha06}
Masahito Hayashi.
\newblock ``Parallel treatment of estimation of {SU}(2) and phase estimation''.
\newblock
  \href{https://dx.doi.org/https://doi.org/10.1016/j.physleta.2006.01.043}{Physics
  Letters A {\bf 354}, 183--189}~(2006).

\bibitem{PhysRevResearch.6.043171}
Shoukang Chang, Yunbin Yan, Lu~Wang, Wei Ye, Xuan Rao, Huan Zhang, Liqing
  Huang, Mengmeng Luo, Yuetao Chen, Qiang Ma, and Shaoyan Gao.
\newblock ``Global quantum thermometry based on the optimal biased bound''.
\newblock \href{https://dx.doi.org/10.1103/PhysRevResearch.6.043171}{Phys. Rev.
  Res. {\bf 6}, 043171}~(2024).

\bibitem{PhysRevLett.127.190402}
Jes\'us Rubio, Janet Anders, and Luis~A. Correa.
\newblock ``Global quantum thermometry''.
\newblock \href{https://dx.doi.org/10.1103/PhysRevLett.127.190402}{Phys. Rev.
  Lett. {\bf 127}, 190402}~(2021).

\bibitem{hayashi2016quantum}
Masahito Hayashi.
\newblock ``Quantum information theory: Mathematical foundation''.
\newblock \href{https://dx.doi.org/10.1007/978-3-662-49725-8}{Springer}.
  ~(2016).

\bibitem{nagaoka89selectedpapers}
Hiroshi Nagaoka.
\newblock ``On fisher information of quantum statistical models''.
\newblock In Masahito Hayashi, editor, Asymptotic Theory Of Quantum Statistical
  Inference: Selected Papers.
\newblock \href{https://dx.doi.org/10.1142/9789812563071_0011}{Pages 113--124}.
\newblock World Scientific~(2005).

\bibitem{nagaoka87selectedpapers}
Hiroshi Nagaoka.
\newblock ``On the parameter estimation problem for quantum statistical
  models''.
\newblock In Masahito Hayashi, editor, Asymptotic Theory Of Quantum Statistical
  Inference: Selected Papers.
\newblock \href{https://dx.doi.org/10.1142/9789812563071_0011}{Pages 125--132}.
\newblock World Scientific~(2005).

\bibitem{Giovannetti}
Vittorio Giovannetti, Seth Lloyd, and Lorenzo Maccone.
\newblock ``Quantum-enhanced measurements: Beating the standard quantum
  limit''.
\newblock \href{https://dx.doi.org/10.1126/science.1104149}{Science {\bf 306},
  1330--1336}~(2004).

\bibitem{Giovannetti06}
Vittorio Giovannetti, Seth Lloyd, and Lorenzo Maccone.
\newblock ``Quantum metrology''.
\newblock \href{https://dx.doi.org/10.1103/PhysRevLett.96.010401}{Phys. Rev.
  Lett. {\bf 96}, 010401}~(2006).

\bibitem{Giovannetti11}
Vittorio Giovannetti, Seth Lloyd, and Lorenzo Maccone.
\newblock ``Advances in quantum metrology''.
\newblock Nature Photonics {\bf 5}, 222–229~(2011).
\newblock  url:~\url{https://doi.org/10.1038/nphoton.2011.35}.

\bibitem{Thomas-Peter}
Nicholas Thomas-Peter, Brian~J. Smith, Animesh Datta, Lijian Zhang, Uwe Dorner,
  and Ian~A. Walmsley.
\newblock ``Real-world quantum sensors: Evaluating resources for precision
  measurement''.
\newblock \href{https://dx.doi.org/10.1103/PhysRevLett.107.113603}{Phys. Rev.
  Lett. {\bf 107}, 113603}~(2011).

\bibitem{Jonathan}
Jonathan~A. Jones, Steven~D. Karlen, Joseph Fitzsimons, Arzhang Ardavan,
  Simon~C. Benjamin, G.~Andrew~D. Briggs, and John J.~L. Morton.
\newblock ``Magnetic field sensing beyond the standard quantum limit using
  10-spin noon states''.
\newblock \href{https://dx.doi.org/10.1126/science.1170730}{Science {\bf 324},
  1166--1168}~(2009).

\bibitem{Okamoto_2008}
Ryo Okamoto, Holger~F Hofmann, Tomohisa Nagata, Jeremy~L O'Brien, Keiji Sasaki,
  and Shigeki Takeuchi.
\newblock ``Beating the standard quantum limit: phase super-sensitivity of
  n-photon interferometers''.
\newblock \href{https://dx.doi.org/10.1088/1367-2630/10/7/073033}{New Journal
  of Physics {\bf 10}, 073033}~(2008).

\bibitem{Nagata07}
Tomohisa Nagata, Ryo Okamoto, Jeremy~L. O'Brien, Keiji Sasaki, and Shigeki
  Takeuchi.
\newblock ``Beating the standard quantum limit with four-entangled photons''.
\newblock \href{https://dx.doi.org/10.1126/science.1138007}{Science {\bf 316},
  726--729}~(2007).

\bibitem{hayashi2018resolving}
Masahito Hayashi, Sai Vinjanampathy, and LC~Kwek.
\newblock ``Resolving unattainable cramer--rao bounds for quantum sensors''.
\newblock \href{https://dx.doi.org/10.1088/1361-6455/aaf348}{Journal of Physics
  B: Atomic, Molecular and Optical Physics {\bf 52}, 015503}~(2018).

\bibitem{Paninski}
L.~Paninski.
\newblock ``Estimating entropy on m bins given fewer than m samples''.
\newblock \href{https://dx.doi.org/10.1109/TIT.2004.833360}{IEEE Transactions
  on Information Theory {\bf 50}, 2200--2203}~(2004).

\bibitem{VV}
Gregory Valiant and Paul Valiant.
\newblock ``Estimating the unseen: an n/log (n)-sample estimator for entropy
  and support size, shown optimal via new clts''.
\newblock In Proceedings of the forty-third annual ACM symposium on Theory of
  computing.
\newblock \href{https://dx.doi.org/10.1145/1993636.1993727}{Pages 685--694}.
\newblock ~(2011).

\bibitem{AOTT}
Jayadev Acharya, Alon Orlitsky, Ananda~Theertha Suresh, and Himanshu Tyagi.
\newblock ``The complexity of estimating r{\'e}nyi entropy''.
\newblock In Proceedings of the twenty-sixth annual ACM-SIAM symposium on
  Discrete algorithms.
\newblock \href{https://dx.doi.org/10.1137/1.9781611973730.124}{Pages
  1855--1869}.
\newblock SIAM~(2014).

\bibitem{AISW}
Jayadev Acharya, Ibrahim Issa, Nirmal~V. Shende, and Aaron~B. Wagner.
\newblock ``Estimating quantum entropy''.
\newblock \href{https://dx.doi.org/10.1109/JSAIT.2020.3015235}{IEEE Journal on
  Selected Areas in Information Theory {\bf 1}, 454--468}~(2020).

\bibitem{WZ}
Qisheng Wang and Zhicheng Zhang.
\newblock ``{Time-Efficient Quantum Entropy Estimator via Samplizer}''.
\newblock In Timothy Chan, Johannes Fischer, John Iacono, and Grzegorz Herman,
  editors, 32nd Annual European Symposium on Algorithms (ESA 2024).
\newblock \href{https://dx.doi.org/10.4230/LIPIcs.ESA.2024.101}{Volume 308 of
  Leibniz International Proceedings in Informatics (LIPIcs), pages
  101:1--101:15}.
\newblock Dagstuhl, Germany~(2024). Schloss Dagstuhl -- Leibniz-Zentrum f{\"u}r
  Informatik.

\bibitem{TAD}
Mankei Tsang, Francesco Albarelli, and Animesh Datta.
\newblock ``Quantum semiparametric estimation''.
\newblock \href{https://dx.doi.org/10.1103/PhysRevX.10.031023}{Phys. Rev. X
  {\bf 10}, 031023}~(2020).

\bibitem{H24}
Masahito Hayashi.
\newblock ``Measuring quantum relative entropy with finite-size effect''.
\newblock \href{https://dx.doi.org/10.22331/q-2025-05-05-1725}{{Quantum} {\bf
  9}, 1725}~(2025).

\bibitem{H02}
Masahito Hayashi.
\newblock ``Optimal sequence of quantum measurements in the sense of {S}tein's
  lemma in quantum hypothesis testing''.
\newblock \href{https://dx.doi.org/10.1088/0305-4470/35/50/307}{Journal of
  Physics A: Mathematical and General {\bf 35}, 10759}~(2002).

\bibitem{two3}
Masahito Hayashi.
\newblock ``Comparison between the {C}ramer-{R}ao and the mini-max approaches
  in quantum channel estimation''.
\newblock \href{https://dx.doi.org/10.1007/s00220-011-1239-4}{Commun. Math.
  Phys. {\bf 304}, 689 -- 709}~(2011).

\bibitem{PhysRevLett.128.130502}
Mohammad Mehboudi, Mathias~R. J\o{}rgensen, Stella Seah, Jonatan~B. Brask, Jan
  Ko\l{}ody\ifmmode~\acute{n}\else \'{n}\fi{}ski, and Mart\'{\i}
  Perarnau-Llobet.
\newblock ``Fundamental limits in bayesian thermometry and attainability via
  adaptive strategies''.
\newblock \href{https://dx.doi.org/10.1103/PhysRevLett.128.130502}{Phys. Rev.
  Lett. {\bf 128}, 130502}~(2022).

\bibitem{MH16-8}
Masahito Hayashi and Shun Watanabe.
\newblock ``Information geometry approach to parameter estimation in markov
  chains''.
\newblock \href{https://dx.doi.org/10.1214/15-AOS1420}{The Annals of Statistics
  {\bf 44}, 1495--1535}~(2016).

\bibitem{Hain17-3}
Masahito Hayashi.
\newblock ``Information geometry approach to parameter estimation in hidden
  {M}arkov model''.
\newblock \href{https://dx.doi.org/10.3150/21-BEJ1344}{Bernoulli {\bf 28},
  307--342}~(2022).

\bibitem{PhysRevLett.108.230401}
Mankei Tsang.
\newblock ``Ziv-zakai error bounds for quantum parameter estimation''.
\newblock \href{https://dx.doi.org/10.1103/PhysRevLett.108.230401}{Phys. Rev.
  Lett. {\bf 108}, 230401}~(2012).

\bibitem{Rubio_2018}
Jes{\'u}s Rubio, Paul Knott, and Jacob Dunningham.
\newblock ``Non-asymptotic analysis of quantum metrology protocols beyond the
  cram{\'e}r-rao bound''.
\newblock \href{https://dx.doi.org/10.1088/2399-6528/aaa234}{Journal of Physics
  Communications {\bf 2}, 015027}~(2018).

\bibitem{Braunstein_1992}
S~L Braunstein.
\newblock ``How large a sample is needed for the maximum likelihood estimator
  to be approximately gaussian?''.
\newblock \href{https://dx.doi.org/10.1088/0305-4470/25/13/027}{Journal of
  Physics A: Mathematical and General {\bf 25}, 3813}~(1992).

\bibitem{PhysRevResearch.6.L032048}
Zhao-Yi Zhou, Jing-Tao Qiu, and Da-Jian Zhang.
\newblock ``Strict hierarchy of optimal strategies for global estimations:
  Linking global estimations with local ones''.
\newblock \href{https://dx.doi.org/10.1103/PhysRevResearch.6.L032048}{Phys.
  Rev. Res. {\bf 6}, L032048}~(2024).

\bibitem{PhysRevA.104.052214}
Julia Boeyens, Stella Seah, and Stefan Nimmrichter.
\newblock ``Uninformed bayesian quantum thermometry''.
\newblock \href{https://dx.doi.org/10.1103/PhysRevA.104.052214}{Phys. Rev. A
  {\bf 104}, 052214}~(2021).

\bibitem{Nagaoka23}
Hiroshi Nagaoka and Akio Fujiwara.
\newblock ``Autoparallelity of quantum statistical manifolds in the light of
  quantum estimation theory''~(2023).
\newblock  \href{http://arxiv.org/abs/2307.03431}{arXiv:2307.03431}.

\bibitem{anglin2002bose}
James~R Anglin and Wolfgang Ketterle.
\newblock ``Bose-{E}instein condensation of atomic gases''.
\newblock \href{https://dx.doi.org/10.1038/416211a}{Nature {\bf 416},
  211--218}~(2002).

\bibitem{johnsson2020geometric}
Mattias~T. Johnsson, Nabomita~Roy Mukty, Daniel Burgarth, Thomas Volz, and
  Gavin~K. Brennen.
\newblock ``Geometric pathway to scalable quantum sensing''.
\newblock \href{https://dx.doi.org/10.1103/PhysRevLett.125.190403}{Phys. Rev.
  Lett. {\bf 125}, 190403}~(2020).

\bibitem{ouyang2022quantum}
Yingkai Ouyang and Gavin~K Brennen.
\newblock ``Finite-round quantum error correction on symmetric quantum
  sensors''~(2022).
\newblock  \href{http://arxiv.org/abs/2212.06285}{arXiv:arxiv:2212.06285}.

\bibitem{PhysRevA.85.022321}
Philipp Hyllus, Wies\l{}aw Laskowski, Roland Krischek, Christian Schwemmer,
  Witlef Wieczorek, Harald Weinfurter, Luca Pezz\'e, and Augusto Smerzi.
\newblock ``Fisher information and multiparticle entanglement''.
\newblock \href{https://dx.doi.org/10.1103/PhysRevA.85.022321}{Phys. Rev. A
  {\bf 85}, 022321}~(2012).

\bibitem{tothPhysRevA.85.022322}
G\'eza T\'oth.
\newblock ``Multipartite entanglement and high-precision metrology''.
\newblock \href{https://dx.doi.org/10.1103/PhysRevA.85.022322}{Phys. Rev. A
  {\bf 85}, 022322}~(2012).

\bibitem{lucke2011twin}
Bernd L{\"u}cke, Manuel Scherer, Jens Kruse, Luca Pezz{\'e}, Frank
  Deuretzbacher, Phillip Hyllus, Oliver Topic, Jan Peise, Wolfgang Ertmer, Jan
  Arlt, et~al.
\newblock ``Twin matter waves for interferometry beyond the classical limit''.
\newblock \href{https://dx.doi.org/10.1126/science.1208798}{Science {\bf 334},
  773--776}~(2011).

\bibitem{halfdicke}
Yi-Quan Zou, Ling-Na Wu, Qi~Liu, Xin-Yu Luo, Shuai-Feng Guo, Jia-Hao Cao,
  Meng~Khoon Tey, and Li~You.
\newblock ``Beating the classical precision limit with spin-1 {D}icke states of
  more than 10,000 atoms''.
\newblock \href{https://dx.doi.org/10.1073/pnas.1715105115}{Proceedings of the
  National Academy of Sciences {\bf 115}, 6381--6385}~(2018).

\bibitem{DARIANO1998103}
G.M D'Ariano, C~Macchiavello, and M.F Sacchi.
\newblock ``On the general problem of quantum phase estimation''.
\newblock
  \href{https://dx.doi.org/https://doi.org/10.1016/S0375-9601(98)00702-6}{Physics
  Letters A {\bf 248}, 103--108}~(1998).

\bibitem{Macchiavello03}
Chiara Macchiavello.
\newblock ``Optimal estimation of multiple phases''.
\newblock \href{https://dx.doi.org/10.1103/PhysRevA.67.062302}{Phys. Rev. A
  {\bf 67}, 062302}~(2003).

\bibitem{Gebhart}
Valentin Gebhart, Augusto Smerzi, and Luca Pezz\`e.
\newblock ``Bayesian quantum multiphase estimation algorithm''.
\newblock \href{https://dx.doi.org/10.1103/PhysRevApplied.16.014035}{Phys. Rev.
  Appl. {\bf 16}, 014035}~(2021).

\bibitem{Wojciech22}
Wojciech G\'orecki and Rafa\l{} Demkowicz-Dobrza\ifmmode~\acute{n}\else
  \'{n}\fi{}ski.
\newblock ``Multiple-phase quantum interferometry: Real and apparent gains of
  measuring all the phases simultaneously''.
\newblock \href{https://dx.doi.org/10.1103/PhysRevLett.128.040504}{Phys. Rev.
  Lett. {\bf 128}, 040504}~(2022).

\bibitem{Hayashi_2022}
Masahito Hayashi, Zi-Wen Liu, and Haidong Yuan.
\newblock ``Global heisenberg scaling in noisy and practical phase
  estimation''.
\newblock \href{https://dx.doi.org/10.1088/2058-9565/ac5d7e}{Quantum Science
  and Technology {\bf 7}, 025030}~(2022).

\bibitem{Luis}
A.~Luis and J.~Pe\ifmmode~\check{r}\else \v{r}\fi{}ina.
\newblock ``Optimum phase-shift estimation and the quantum description of the
  phase difference''.
\newblock \href{https://dx.doi.org/10.1103/PhysRevA.54.4564}{Phys. Rev. A {\bf
  54}, 4564--4570}~(1996).

\bibitem{Buzek}
V.~Bu\ifmmode~\check{z}\else \v{z}\fi{}ek, R.~Derka, and S.~Massar.
\newblock ``Optimal quantum clocks''.
\newblock \href{https://dx.doi.org/10.1103/PhysRevLett.82.2207}{Phys. Rev.
  Lett. {\bf 82}, 2207--2210}~(1999).

\bibitem{Higgins2007}
B.~L. Higgins, D.~W. Berry, S.~D. Bartlett, H.~M. Wiseman, and G.~J. Pryde.
\newblock ``Entanglement-free heisenberg-limited phase estimation''.
\newblock \href{https://dx.doi.org/10.1038/nature06257}{Nature {\bf 450},
  393–396}~(2007).

\bibitem{Kitaev}
Alexei~Y. Kitaev.
\newblock ``Quantum measurements and the abelian stabilizer problem''.
\newblock Electron. Colloquium Comput. Complex.{\bf {TR96-003}}~(1996).
\newblock  \href{http://arxiv.org/abs/TR96-003}{arXiv:TR96-003}.

\bibitem{Cleve}
R.~Cleve, A.~Ekert, C.~Macchiavello, and M.~Mosca.
\newblock ``Quantum algorithms revisited''.
\newblock Proc. R. Soc. Lond. A. {\bf 454}, 339–354~(1998).
\newblock  url:~\url{http://doi.org/10.1098/rspa.1998.0164}.

\bibitem{Imai_2009}
Hiroshi Imai and Masahito Hayashi.
\newblock ``Fourier analytic approach to phase estimation in quantum systems''.
\newblock \href{https://dx.doi.org/10.1088/1367-2630/11/4/043034}{New Journal
  of Physics {\bf 11}, 043034}~(2009).

\bibitem{HOLEVO1979385}
A.S. Holevo.
\newblock ``Covariant measurements and uncertainty relations''.
\newblock
  \href{https://dx.doi.org/https://doi.org/10.1016/0034-4877(79)90072-7}{Reports
  on Mathematical Physics {\bf 16}, 385--400}~(1979).

\bibitem{Group2}
Masahito Hayashi.
\newblock ``Group representation for quantum theory''.
\newblock \href{https://dx.doi.org/10.1007/978-3-319-44906-7}{Springer}.
  ~(2017).

\bibitem{ouyang2019permutation}
Yingkai Ouyang and Rui Chao.
\newblock ``Permutation-invariant constant-excitation quantum codes for
  amplitude damping''.
\newblock \href{https://dx.doi.org/10.1109/TIT.2019.2956142}{IEEE Transactions
  on Information Theory {\bf 66}, 2921--2933}~(2019).

\bibitem{erhard2020advances}
Manuel Erhard, Mario Krenn, and Anton Zeilinger.
\newblock ``Advances in high-dimensional quantum entanglement''.
\newblock \href{https://dx.doi.org/10.1038/s42254-020-0193-5}{Nature Reviews
  Physics {\bf 2}, 365--381}~(2020).

\bibitem{Alfredo-Luis_2001}
Alfredo Luis.
\newblock ``Quantum mechanics as a geometric phase: phase-space
  interferometers''.
\newblock \href{https://dx.doi.org/10.1088/0305-4470/34/37/317}{Journal of
  Physics A: Mathematical and General {\bf 34}, 7677}~(2001).

\bibitem{zanardi-PhysRevA.65.032327}
Xiaoguang Wang and Paolo Zanardi.
\newblock ``Simulation of many-body interactions by conditional geometric
  phases''.
\newblock \href{https://dx.doi.org/10.1103/PhysRevA.65.032327}{Phys. Rev. A
  {\bf 65}, 032327}~(2002).

\bibitem{H97}
Masahito Hayashi.
\newblock ``A linear programming approach to attainable cram\'{e}r-rao type
  bounds and randomness condition''~(1997).
\newblock
  \href{http://arxiv.org/abs/quant-ph/9704044}{arXiv:quant-ph/9704044}.

\bibitem{nagaoka89-2}
Hiroshi Nagaoka.
\newblock ``A new approach to {C}ram\'{e}r-{R}ao bounds for quantum state
  estimation''.
\newblock In Masahito Hayashi, editor, Asymptotic Theory Of Quantum Statistical
  Inference: Selected Papers.
\newblock \href{https://dx.doi.org/10.1142/9789812563071_0009}{Pages 100--112}.
\newblock World Scientific~(2005).

\bibitem{AN}
Shun'ichi Amari and Hiroshi Nagaoka.
\newblock ``Methods of information geometry''.
\newblock \href{https://dx.doi.org/10.1007/978-4-431-55978-8}{Volume 191}.
\newblock American Mathematical Soc. ~(2000).

\bibitem{MH16-9}
Masahito Hayashi.
\newblock ``Fourier analytic approach to quantum estimation of group action''.
\newblock \href{https://dx.doi.org/10.1007/s00220-016-2738-0}{Communications in
  Mathematical Physics {\bf 347}, 3--82}~(2016).

\bibitem{F02}
Akio Fujiwara.
\newblock ``Estimation of {SU}(2) operation and dense coding: An information
  geometric approach''.
\newblock \href{https://dx.doi.org/10.1103/PhysRevA.65.012316}{Phys. Rev. A
  {\bf 65}, 012316}~(2001).

\bibitem{IF07}
Hiroshi Imai and Akio Fujiwara.
\newblock ``Geometry of optimal estimation scheme for {SU}({D}) channels''.
\newblock \href{https://dx.doi.org/10.1088/1751-8113/40/16/009}{Journal of
  Physics A: Mathematical and Theoretical {\bf 40}, 4391}~(2007).

\bibitem{mh-photon-constant}
Masahito Hayashi.
\newblock ``Phase estimation with photon number constraint''.
\newblock \href{https://dx.doi.org/10.2201/NiiPi.2011.8.9}{Progress in
  InformaticsPages 81--87}~(2011).

\bibitem{KG}
Jonas Kahn and M{\u{a}}d{\u{a}}lin Gu{\c{t}}{\u{a}}.
\newblock ``Local asymptotic normality for finite dimensional quantum
  systems''.
\newblock \href{https://dx.doi.org/10.1007/s00220-009-0787-3}{Communications in
  Mathematical Physics {\bf 289}, 597--652}~(2009).

\bibitem{Kahn}
Jonas Kahn.
\newblock ``Quantum local asymptotic normality and other questions of quantum
  statistics''.
\newblock Leiden University. ~(2008).
\newblock  url:~\url{https://hdl.handle.net/1887/12956}.

\bibitem{uwePhysRevLett.132.240803}
Jae-Gyun Baak and Uwe~R. Fischer.
\newblock ``Self-consistent many-body metrology''.
\newblock \href{https://dx.doi.org/10.1103/PhysRevLett.132.240803}{Phys. Rev.
  Lett. {\bf 132}, 240803}~(2024).

\bibitem{pitomograph-PhysRevLett.105.250403}
G.~T\'oth, W.~Wieczorek, D.~Gross, R.~Krischek, C.~Schwemmer, and
  H.~Weinfurter.
\newblock ``Permutationally invariant quantum tomography''.
\newblock \href{https://dx.doi.org/10.1103/PhysRevLett.105.250403}{Phys. Rev.
  Lett. {\bf 105}, 250403}~(2010).

\bibitem{moroder2012permutationally}
Tobias Moroder, Philipp Hyllus, G{\'e}za T{\'o}th, Christian Schwemmer,
  Alexander Niggebaum, Stefanie Gaile, Otfried G{\"u}hne, and Harald
  Weinfurter.
\newblock ``Permutationally invariant state reconstruction''.
\newblock \href{https://dx.doi.org/10.1088/1367-2630/14/10/105001}{New Journal
  of Physics {\bf 14}, 105001}~(2012).

\bibitem{PhysRevLett.113.040503}
Christian Schwemmer, G\'eza T\'oth, Alexander Niggebaum, Tobias Moroder, David
  Gross, Otfried G\"uhne, and Harald Weinfurter.
\newblock ``Experimental comparison of efficient tomography schemes for a
  six-qubit state''.
\newblock \href{https://dx.doi.org/10.1103/PhysRevLett.113.040503}{Phys. Rev.
  Lett. {\bf 113}, 040503}~(2014).

\end{thebibliography}

\end{document}